\documentclass{sig-alternate-10}

\usepackage[utf8]{inputenc}
 \usepackage{amsmath}
 \usepackage{amssymb}
\usepackage{bm}
\usepackage{mathabx}
\usepackage{mathtools}
\usepackage{stmaryrd}
\usepackage{algorithm}
\usepackage[noend]{algpseudocode}
\usepackage{comment}
\usepackage[textwidth=2cm,textsize=small]{todonotes}
\usepackage{caption}
\usepackage{subcaption}
\usepackage{tikz}
\usetikzlibrary{automata, arrows.meta, graphs, graphs.standard, positioning}

\usepackage{enumitem}
\newlist{todolist}{itemize}{2}
\setlist[todolist]{label=\rlap{$\Box$}{\phantom{\cmark}}}
\usepackage{pifont}
\newcommand{\cmark}{\ding{51}}%
\newcommand{\PTIME}{\textnormal{\textsf{PTIME}}}

\newcommand{\PSPACE}{\textnormal{\textsf{PSPACE}}}

\newcommand{\ALOGSPACE}{\textnormal{\textsf{ALOGSPACE}}}
\newcommand{\NLOGSPACE}{\textnormal{\textsf{NLOGSPACE}}}
\newcommand{\coNLOGSPACE}{\textnormal{\textsf{coNLOGSPACE}}}
\newcommand{\NP}{\textnormal{\textsf{NP}}}
\newcommand{\coNP}{\textnormal{\textsf{coNP}}}
\newcommand{\FPT}{\textnormal{\textsf{FPT}}}
\newcommand{\pitwop}{\Pi_2^p}

\newcommand{\ONETHREESAT}{\textnormal{\textsf{1-IN-3-SAT}}}

\newtheorem{theorem}{Theorem}[section]
\newtheorem{example}[theorem]{Example}

\newtheorem{proposition}[theorem]{Proposition}

\newtheorem{lemma}[theorem]{Lemma}

\newcommand{\cA}{\mathcal{A}}
\newcommand{\eol}{\rotatebox[origin=l]{180}{\ensuremath\Rsh}}

\newcommand{	\VV}{\mathcal{V}}
\newcommand{	\OO}{\mathcal{O}}

\newcommand{	\sub}{\text{span}}

\newcommand{	\var}{\text{var}}
\newcommand{	\ivar}{\text{ivar}}

\newcommand{	\dom}{\text{dom}}
\newcommand{\sem}[1]{\llbracket#1\rrbracket}

\newcommand{\semp}[1]{[#1]}
\newcommand{\semg}[1]{\llbracket #1 \rrbracket_{d}}
\newcommand{\semd}[1]{\sem{#1}_{d}}

\newcommand{\sempd}[1]{\semp{#1}_{d}}

\newcommand{\trans}[2][]{\raisebox{-1pt}[10pt][0pt]{$\overset{#2}{\underset{^{#1}}{\raisebox{0pt}[3pt][0pt]{$\relbar\mspace{-8mu}\rightarrow$}}}$}}

\newcommand{\VA}{\ensuremath{\mathrm{VA}}}

\newcommand{\seqVA}{\ensuremath{\mathrm{seqVA}}}
\newcommand{\Stk}{\ensuremath{\mathrm{stk}}}
\newcommand{\VASet}{\ensuremath{\mathrm{VA_{set}}}}
\newcommand{\VAStk}{\ensuremath{\mathrm{VA_{\Stk}}}}

\newcommand{\PUStk}{\ensuremath{\mathrm{PU_{stk}}}}

\DeclareMathOperator{\runs}{Runs}
\DeclareMathOperator{\aruns}{ARuns}

\newcommand{\RGX}{\ensuremath{\mathrm{RGX}}}
\newcommand{\fRGX}{\ensuremath{\mathrm{funcRGX}}}

\newcommand{\SRGX}{\ensuremath{\mathrm{spanRGX}}}
\newcommand{\seqRGX}{\ensuremath{\mathrm{seqRGX}}}

\newcommand{\Vars}{\mathcal{V}}

\newcommand{\Var}{\operatorname{var}}

\newcommand{\Open}[1]{#1\mkern-3mu\vdash}
\newcommand{\Close}[1]{\dashv\mkern-3mu#1}

\newcommand{\nemp}{\text{\sc NonEmp}}
\newcommand{\mdcheck}{\text{\sc ModelCheck}}
\newcommand{\EVAL}{\text{\sc Eval}}
\newcommand{\sat}{\text{\sc Sat}}
\newcommand{\containment}{\text{\sc Containment}}

\newcommand{\LL}{{\cal L}}

\newcommand{\spa}{\text{\textvisiblespace}}

\DeclareMathOperator{\Perm}{Perm}
\DeclareMathOperator{\Op}{Op}

\newcommand{\secref}[1]{\lowercase{\ref{#1}}}

\title{Document Spanners for Extracting Incomplete Information: Expressiveness and Complexity}

\author{
Francisco Maturana\\
       \affaddr{PUC Chile}\\
       \email{fjmaturana@uc.cl}
\and
Cristian Riveros\\
       \affaddr{PUC Chile}\\
       \email{cristian.riveros@uc.cl}
\and
Domagoj Vrgo\v{c} \\
       \affaddr{PUC Chile}\\
       \email{dvrgoc@ing.puc.cl}
}

\begin{document}
\maketitle

\sloppy

%
%

\begin{abstract}
	
Rule-based information extraction has lately received a fair amount of attention from the database community, with several languages appearing in the last few years. Although information extraction systems are intended to deal with semistructured data, all language proposals introduced so far are designed to output relations, thus making them incapable of handling incomplete information. To remedy the situation, we propose to extend information extraction languages with the ability to use mappings, thus allowing us to work with documents which have missing or optional parts. Using this approach, we simplify the semantics of regex formulas and extraction rules, two previously defined methods for extracting information, extend them with the ability to handle incomplete data, and study how they compare in terms of expressive power. We also study computational properties of these languages, focusing on the query enumeration problem, as well as satisfiability and containment.
	
\end{abstract}



\section{Introduction}
With the abundance of different formats arising in practice these days, there is a great need for methods extracting singular pieces of data from a variety of distinct files. This process, known as information extraction, or IE for short, is particularly prevalent in big corporations that manage systems of increasing complexity which need to incorporate data coming from different sources. As a result, a number of systems supporting the extraction of information from text-like data have appeared throughout the years \cite{KLRRVZ08,Cunningham02,SDNR07}, and the topic received a substantial coverage in research literature (see \cite{Kimelfeld14} for a good survey).

Historically, there have been two main approaches to information extraction: the {\em statistical} approach utilising machine-learning methods, and the {\em rule-based} approach utilising traditional finite-language methods. The latter approach has lately enjoyed a great amount of coverage in the database literature \cite{FKRV15,FKRV14,freydenberger2016document,AMRV16} showing interesting connections with logic, automata, or datalog-based languages. Furthermore, as argued by \cite{Kimelfeld14,ChiticariuLR13}, due to their simplicity and ease of use, rule-based systems also seem to be more prevalent in the industrial solutions.


Generally, most rule-based IE frameworks view documents as strings, which is a natural assumption for many formats in use today (e.g.\ plain text, CSV files, JSON documents). The information we want to extract is then represented by {\em spans}, which are simply intervals inside the string representing our document; that is, a span specifies a substring (i.e.\ the data) plus its starting and ending position inside the document.
The process of extracting information can then be captured by the notion of {\em document spanners}, which are simply operators that transforms an arbitrary string, i.e.\ a document, into a relation containing spans over this string.


In order to specify basic document spanners, most rule-based IE frameworks use some form of regular expressions with capture variables~\cite{CM99,FKRV15,AMRV16}. Perhaps the best example of this are the {\em regex formulas} of \cite{FKRV15}, which form the basis of IBM's commercial IE tool SystemT \cite{KLRRVZ08}. The main idea behind such expressions is quite natural: to use regular expressions in order to locate the span that is to be extracted, and then use variables to store this span. Once spans have been extracted using regular-like expressions, most IE frameworks allow combining them and controlling their structure through a variety of different methods. For instance, \cite{FKRV15} permits manipulating spans extracted by regex formulas using algebraic operations, while \cite{AMRV16} and \cite{SDNR07} deploy Datalog-like rules to define relations over spans.

And while several proposals for information extraction frameworks have appeared throughout the years \cite{CM99,SDNR07,FKRV15,AMRV16}, each of them offering significant advantages for the specific context they were designed to operate in, we believe that there are still some challenges not addressed by these languages, nor by the research literature as a whole. We next identify several such challenges which, once resolved, could lead to a better understanding of the information extraction process.

First, as already mentioned, the majority of methods for defining document spanners view information extraction as a process that defines a relation over spans. For example, in regex formulas of \cite{FKRV15}, all variables must be assigned in order to produce an output tuple, and a similar thing happens with extraction rules of \cite{AMRV16}. However, in practice we are often working with documents which have missing information or optional parts, and would therefore like to maximise the amount of information we extract. To illustrate this, consider a CSV file\footnote{CSV, or comma separated values, is a simple table-like format storing information separated by commas and new lines.} containing land registry records about buying and selling property. In Table \ref{tab-csv} we give a few rows of such a document, where $\spa$ represents space and $\eol$ the new line symbol. Some sellers in this file have an additional field which contains the amount of tax they paid when selling the property. If we are extracting information about sellers (for instance their names) from such a file, we would then like to also include the tax information when the latter is available. Unfortunately, most previous proposals (see e.g. \cite{FKRV15,AMRV16}) are not well suited for this task, as they require all the variables to be assigned in order to produce an output, thus causing us to miss some of the desired data.


\begin{table}
\begin{center}
{
\def\arraystretch{1.2}
\renewcommand{\familydefault}{\ttdefault}\normalfont
\begin{tabular}{l}
Seller:$\spa$John,$\spa$ID75\eol\\
{Buyer:}$\spa$Marcelo,$\spa$ID832,$\spa$P78\eol\\
Seller:$\spa$Mark,$\spa$ID7,$\spa$\$35,000\eol \vspace{-5pt}\\
\hspace{55pt}$\vdots$
\end{tabular}
}
\end{center}
\vspace*{-10pt}
\caption{Part of a CSV document containing information about buying and selling property.}
\label{tab-csv}
\end{table}

%
%
%

Another drawback of previous approaches to IE is that there is no agreement on the correct way to define the semantics of basic document spanners. For instance, up to date there is no fully declarative semantics for regex formulas of \cite{FKRV15}, and their meaning is usually given in a procedural manner: either through syntactic parse trees \cite{FKRV15}, or using automata \cite{F17}. Similarly, approaches such as \cite{AMRV16} allow assigning arbitrary spans to variables when these are not matched against the document, thus potentially extracting undesired, or even incorrect, information.

Finally, not much is known about how different information extraction frameworks compare in terms of expressive power, nor about their computational properties. For instance, although there is some work on evaluating specific IE languages \cite{freydenberger2016document,AMRV16,F17,FKP17}, we still do not have a good idea of which decision problems faithfully model the process of computing the (potentially exponential) output of the information extraction process, nor do we understand the complexity of the main static tasks associated with IE languages.

\smallskip
\noindent
\textbf{Contributions.} In order to alleviate some of the above issues, in this paper we propose to redefine the semantics of several previously introduced IE languages by making them output {\em mappings} in place of relations. This will not only allow us to capture incomplete information by making our spanners output partial mappings when some data is not available, but will also lend itself to defining a simple declarative semantics for multiple IE languages. This will then allow us to compare these languages in terms of expressive power, and make it easier to understand their computational properties such as query enumeration and query containment.

In particular, in what follows we will consider the regex formulas of \cite{FKRV15}, their automata analogue called variable-set automata \cite{FKRV15}, and extraction rules of \cite{AMRV16}. We first extend these formalism with the ability to output mappings, thus making them capable of handling incomplete information, and give a simple inductive definition of their semantics. As sanity check we then show that this new semantics indeed subsumes the previous proposals of \cite{FKRV15} and \cite{AMRV16}, while at the same time allowing for simple inductive proofs based on the expression syntax, and that the connections between regex formulas and variable-set automata established in \cite{FKRV15} are preserved when using mappings\footnote{Note that in this paper we do not consider the content operator of \cite{AMRV16}, nor the string selection of \cite{FKRV15}, since these do not directly extract information, but rather compare two pieces of existing data.}. Next, we compare the regex formulas of \cite{FKRV15} and extraction rules of \cite{AMRV16} in terms of expressive power. Here we show that, while the two approaches are generally incomparable, one can restrict and simplify extraction rules in a non trivial manner in order to obtain a class equivalent to regex. 

We also study the combined complexity of evaluating extraction expressions over documents. Here we isolate a decision problem which, once solved efficiently, would allow us to {\em enumerate all mappings} an expression outputs when matched to a document. Since the size of the answer is potentially exponential here, our objective is to obtain a polynomial delay algorithm \cite{johnson1988generating}; an enumeration algorithm that takes polynomial time between each output. As we show, this is generally not possible, but we do isolate well-behaved fragments of the three extraction languages we consider here, all of them based on the idea of sequentially extracting the data. We also analyse the evaluation problem parametrised by the number of variables and show that the problem is \emph{fixed parameter tractable}~\cite{FlumG06} for all expressions and automata models we consider.

Finally, we study static analysis of IE languages, focusing on satisfiability and containment. While satisfiability is $\NP$-hard for unrestricted languages, the sequentiality restriction introduced when studying evaluation allows us to solve the problem efficiently. On the other hand, containment is bound to be $\PSPACE$-hard, since all of our IE formalisms contain regular expressions, with a matching upper bound giving us completeness for the class. Since one way to lower this bound for regular languages is to consider deterministic models, we show how determinism can be introduced to IE languages and study how it affects the complexity.

\smallskip
\noindent{\bf Organisation.} We define documents, spans and mappings in Section~\ref{sec-prelim}. Expressions, automata and rules for extracting incomplete information are introduced in Section~\ref{sec-unifying}. Expressiveness of our languages is studied in Section~\ref{sec:comp}, and the complexity of their evaluation in Section~\ref{sec:evaluation}. We then tackle static analysis in Section~\ref{sec:complexity} and conclude in Section~\ref{sec-conc}. Due to space limitations most of the proofs are deferred to the appendix.

\section{Preliminaries}
\label{sec-prelim}

\noindent{\bf Documents and spans.} Let $\Sigma$ be a finite alphabet. A document $d$, from which we will extract information, is
a string over $\Sigma$. We define the length of $d$,
denoted by $|d|$, as the length of this string. 
As done in  previous approaches \cite{FKRV15,AMRV16}, we use the notion of a {\em span} to capture the part of a document $d$ that we wish to extract. Formally, a span $p$ of a document $d$ is a pair $(i,j)$ such that $1 \leq i \leq j \leq |d|+1$, where $|d|$ is the length of the string $d$. Intuitively, $p$ represents a continuous region of the document $d$, whose content is the infix of $d$ between positions $i$ and $j-1$. The set of all spans associated with a document $d$, denoted $\sub(d)$, is then defined as the set $\{ (i,j) \mid i,j \in \{1, \ldots, |d|+1\} \text{ and } i \leq j\}$. Every span $p = (i,j)$ of $d$ has an associated content, which is denoted by $d(p)$ or $d(i,j)$, and is defined as the substring of $d$ from position $i$ to position $j-1$.
Notice that if $i = j$, then $d(p) = d(i,j) = \varepsilon$. Given two spans $s_1 = (i_1, j_1)$ and $s_2 = (i_2, j_2)$, if $j_1 = i_2$ then their concatenation is equal to $(i_1, j_2)$ and it is denoted $s_1 \cdot s_2$.

As an example, consider the following document $d_0$, where the positions are enumerated and $\spa$ denotes the white space character:

\begin{center}
{
\setlength{\tabcolsep}{2pt}
  \renewcommand{\familydefault}{\ttdefault}\normalfont
  \def\arraystretch{0.8}
\begin{tabular}{cccccccccccccccccccccc}
\small{I} & \small{n} & \small{f} & \small{o} & \small{r} & \small{m} & \small{a} & \small{t} & \small{i} & \small{o} & \small{n} & $\spa$ & \small{e} & \small{x} & \small{t} & \small{r} & \small{a} & \small{c} & \small{t} & \small{i} & \small{o} & \small{n}\\
\tiny{1} & \tiny{2} & \tiny{3} & \tiny{4} & \tiny{5} & \tiny{6} & \tiny{7} & \tiny{8} &  \tiny{9} & \tiny{10} & \tiny{11} & \tiny{12} & \tiny{13} & \tiny{14} & \tiny{15} & \tiny{16} & \tiny{17} & \tiny{18} & \tiny{19} & \tiny{20} & \tiny{21} & \tiny{22} 
\end{tabular}
}
\end{center}
Here the length of $d_0$ is 22 and the span $p_0=(1,23)$ corresponds to the entire document. On the other hand, the span $p_1=(1,12)$ corresponds to the first word of our document and its content $d(p_1)=d(1,12)$ equals the string \texttt{Information}. Similarly, for the span $p_2=(13,23)$ we have that $d(p_2)=\texttt{extraction}$, i.e.\ it spans the second word of our document.

\smallskip

\noindent{\bf Mappings.} In the introduction we argued that the traditional approaches to information extraction that store spans into relations might be somewhat limited when we are processing documents which contain incomplete information. Therefore to overcome these issues, we define the process of extracting information from a document $d$ as if we were defining a partial function from a set of variables to the spans of $d$.
The use of partial functions for managing optional information has been considered before, for example, in the context of the Semantic Web~\cite{PAG09}.
Formally, let \(\VV\) be a set of variables disjoint from $\Sigma$. For a document $d$, a \textit{mapping} is a partial function from the set of variables \(\VV\) to \(\sub(d)\). The \textit{domain} \(\dom(\mu)\) of a mapping \(\mu\) is the set of variables for which \(\mu\) is defined. For instance, if we consider the document $d_0$ above, then the mapping $\mu_0$ which assigns the span $p_1$ to the variable $x$ and leaves all other variables undefined, extracts the first word from $d_0$.

Two mappings \(\mu_1\) and \(\mu_2\) are said to be \textit{compatible} (denoted
\(\mu_1 \sim \mu_2\)) if \(\mu_1(x) = \mu_2(x)\) for every \(x\) in
\(\dom(\mu_1) \cap \dom(\mu_2)\). If \(\mu_1 \sim \mu_2\), then \(\mu_1 \cup
\mu_2\) denotes the mapping that results from extending \(\mu_1\) with the
values from \(\mu_2\) on all the variables in \(\dom(\mu_2) \setminus
\dom(\mu_1)\). The \textit{empty mapping}, denoted by \(\emptyset\), is the
mapping such that \(\dom(\emptyset) = \emptyset\). Similarly, \([x \to s]\)
denotes the mapping that only defines the value of variable \(x\) and assigns it
to be the span~\(s\).
The \textit{join} of two set of mappings \(M_1\) and \(M_2\) is defined as
follows:
\begin{multline*}
  M_1 \Join M_2 = \{\mu_1 \cup \mu_2 \mid \mu_1 \in M_1 \text{ and }
  \mu_2 \in M_2 \\ \text{ such that } \mu_1 \sim \mu_2\}.
\end{multline*}

Finally, we say that a mapping \(\mu\) is \textit{hierarchical} if for every
\(x, y \in \dom(\mu)\), either: \(\mu(x)\) is contained in \(\mu(y)\),
\(\mu(y)\) is contained in \(\mu(x)\), or \(\mu(x)\) and \(\mu(y)\)
are disjoint. Similarly, a set of mappings is said to be hierarchical if it only contains hierarchical mappings.

\section{Extracting incomplete information}
\label{sec-unifying}

In this section we introduce three different mechanisms for extracting data: regex formulas \cite{FKRV15}, variable-set automata \cite{FKRV15}, and extraction rules \cite{AMRV16}; and redefine their semantics in order to support incomplete information. We do this by allowing them to output mappings in place of relations, which makes it possible to provide a simple uniform semantics for different IE approaches proposed in the literature. 


\subsection{Extracting information using RGX}\label{ss-regex}

Most previous approaches to IE \cite{FKRV15,SDNR07,Soderland99,AMRV16} use some form of regular expressions with capture variables in order to obtain the desired spans. Intuitively, in such expressions we use ordinary regular languages to move through our document, thus jumping to the start of a span that we want to capture. The variables are then used to store the desired span, with further subexpressions controlling the shape of the span. Borrowing the syntax from \cite{FKRV15}, we define our core class of extraction expressions, called variable regex, as follows.

Let \(\Sigma\) be a finite alphabet and \(\VV\) a set of variables disjoint with
\(\Sigma\). A \textit{variable regex} (\(\RGX\)) is defined by the following
grammar:
$$
  \gamma \coloneqq \varepsilon \mid
    a \mid
    x\{\gamma\} \mid
  \gamma \cdot \gamma \mid
  \gamma \vee \gamma \mid
  \gamma^{*}
$$
where $a \in \Sigma$ is a letter of the alphabet and $x \in \VV$ is a variable. For a \RGX\ $\gamma$ we define $\Var(\gamma)$ as the set of all  variables occurring in $\gamma$. In what follows we will often refer to variable regex ($\RGX$, resp.) as a regex formula ($\RGX$ formula, resp.).

Just as in the previously introduced IE languages, $\RGX$ use regular expressions to navigate the document, while a subexpression of the form $x\{\gamma\}$ stores a span starting at the current position and matching $\gamma$ into the variable $x$. For example, if we wanted to extract the name of each seller from the document in Table \ref{tab-csv}, we could use the following \RGX\:
$$\Sigma^*\cdot\texttt{Seller:}\spa\cdot x\{(\Sigma-\{\texttt{,}\})^*\}\cdot,\cdot\Sigma^*$$
where $\Sigma$ stands for the disjunction of all the letters of the alphabet, and where we do not use the concatenation $\cdot$ inside words (formally, the string \texttt{Seller:}$\spa$ should be written as the concatenation of each of its symbols). Here the subexpression $\Sigma^*\cdot\texttt{Seller:}\spa$ navigates to the position in our document, where the name of some seller starts. The variable $x$ then stores a string not containing a comma until it reaches the first comma; that is, it stores the full name of our seller. The remainder of the expression then simply matches the rest of the document.

Note that syntactically, our expressions are the same as the ones introduced in \cite{FKRV15}. 
The only explicit difference from \cite{FKRV15} (apart from the semantics -- see below) is that we do not allow the empty language $\emptyset$ in order to make some of the constructions more elegant. Adding this variant would not affect any of the results though.

\smallskip

\noindent{\bf Semantics.} In contrast to \cite{FKRV15}, our semantics views \RGX\ formulas as expressions defining mappings and not only relations. To illustrate how this works, consider again the document in Table \ref{tab-csv}, but now suppose that we want to extract the names of the sellers, and when available, also the amount of tax they paid (recall from the Introduction that not all rows have this information). For this, consider the following \RGX
$$\Sigma^*\cdot\texttt{Seller:}\spa\cdot x\{R'\}\cdot\texttt{,}\cdot R'\cdot (\ \texttt{,}\spa y\{(\Sigma-\{\eol\})^*\} \vee \varepsilon)\cdot \eol \cdot \Sigma^*,$$
where $R'=(\Sigma-\{\texttt{,},\eol\})^*$. Note that this expression extracts the information about the amount of tax paid into the variable $y$ only when this data is present in the document (otherwise it matches $\varepsilon$). This now defines two types of mappings: the first kind will contain only the names of sellers (stored in $x$), while the second kind will contain both the name and the amount of tax paid (stored in $y$) when the latter information is available.

\begin{table}
\begin{align*}
  \semd{\gamma} &= \{ \mu \mid ((1, |d| + 1), \mu) \in \sempd{\gamma} \}\\
  \sempd{\varepsilon} &= \{ (s, \emptyset) \mid s \in \sub(d) \text{ and } d(s)=\varepsilon\}\\
  \sempd{a} &= \{ (s, \emptyset) \mid s \in \sub(d) \text{ and } d(s) = a \}\\
  \sempd{x\{R\}} &= \{ (s, \mu) \mid \exists (s, \mu') \in \sempd{R}:\\
    &\phantom{{}={}} x \not\in \dom(\mu') \text{ and }
    \mu = [x \to s] \cup \mu' \}\\
  \sempd{R_1 \cdot R_2} &= \{ (s, \mu) \mid
    \exists (s_1, \mu_1) \in \sempd{R_1},\\
    &\phantom{{}={}} \exists (s_2, \mu_2) \in \sempd{R_2}: s = s_1 \cdot s_2, \\   
    &\phantom{{}={}} \dom(\mu_1) \cap \dom(\mu_2) = \emptyset, \text{ and }\\
    &\phantom{{}={}} \mu = \mu_1 \cup \mu_2 \}\\
  \sempd{R_1 \vee R_2} &= \sempd{R_1} \cup \sempd{R_2}\\
  \sempd{R^*} &= \sempd{\varepsilon} \cup \sempd{R} \cup \sempd{R^2} \cup
    \sempd{R^3} \cup \cdots\\
\end{align*}
\vspace*{-25pt}
\caption{The semantics \(\semd{\gamma}\) of a \(\RGX\) \(\gamma\) over a document \(d\). Here \(R^2\) is a shorthand for \(R \cdot R\), similarly \(R^3\) for \(R \cdot R \cdot R\), etc.}
\label{tab-semantics}
\end{table}


The full semantics of \RGX\ expressions is defined in Table \ref{tab-semantics}. As explained above, we view our expression $\gamma$ as a way of defining a partial mapping $\mu:\Var(\gamma)\rightharpoonup \sub(d)$. Our semantics has two layers, the first layer $\sempd{\gamma}$ defines which part of a document $d$ a subexpression of $\gamma$ parses, and what is the mapping defined thus far. For instance, the alphabet letter $a$ must match a part of the document equal to $a$ and it defines no mapping. On the other hand, a subexpression of the form $x\{R\}$ assigns to $x$ the span captured by $R$ (and preserves the previous variable assignments). Similarly, in the case of concatenation $R_1\cdot R_2$ we join the mapping defined on the left with the one defined on the right, while imposing that the same variable is not used in both parts (as this would lead to inconsistencies). The second layer of our semantics, $\semd{\gamma}$ then simply gives us the mappings that $\gamma$ defines when matching the entire document. 

Note that in the case of an ordinary regular expression we output the empty mapping (representing $\texttt{TRUE}$) when the expression matches the entire document and empty set (representing $\texttt{FALSE}$) when not, thus making \RGX\  a natural generalisation of ordinary regular expressions with the ability to extract spans. 

As the semantics of some operators might seem somewhat counter intuitive at first, we now explain how the recursive definition works by means of an example.
\begin{example}\label{ex-semantics}
To keep the presentation concise, we will consider the following document $d$:
\begin{center}
{
\setlength{\tabcolsep}{2pt}
  \renewcommand{\familydefault}{\ttdefault}\normalfont
  \def\arraystretch{0.8}
\begin{tabular}{cccccc}
\small{a} & \small{a} & \small{a} & \small{b} & \small{b} & \small{b}\\
\tiny{1} & \tiny{2} & \tiny{3} & \tiny{4} & \tiny{5} & \tiny{6}
\end{tabular}
}
\end{center}
If we consider the expression consisting of a single letter $a$, then the set  $\sempd{a}$ contains precisely three pairs: $((1,2),\emptyset), ((2,3),\emptyset),$ and  $((3,4),\emptyset)$, since the word spelled by each of these spans equals to the letter $a$.

On the other hand, if we consider the expression $x\{a\}$, then $\sempd{x\{a\}}$ contains the above spans, but it also assigns the span to the variable. Namely, $\sempd{x\{a\}}$ consists of the pairs $((i,i+1),\mu_i)$, where $\mu_i(x)=(i,i+1)$ and is undefined otherwise, and where $1\leq i \leq 3$. Notice, however, that $\semd{x\{a\}}$ is empty, since none of the pairs $((i,i+1),\mu_i)$ contains a span representing the entire document $d$.

To illustrate how concatenation works, consider now the expression $x\{a^*\}\cdot y\{b^*\}$. Here $\sempd{a^*}$ contains any span that spells zero or more $a$s, such as for example $((1,4),\emptyset)$, or $((5,5),\emptyset)$. Note that the latter matches $a^*$, as it spells the empty string. Similarly, $\sempd{b^*}$ will contain, amongst others, the pairs $((4,5),\emptyset)$, or $((4,7),\emptyset)$. Because of this we have that $\sempd{x\{a^*\}}$ contains the pair $((1,4),\mu_1)$, with $\mu_1(x)=(1,4)$, while $\sempd{y\{b^*\}}$ contains the pair $((4,7),\mu_2)$, where $\mu_2(y)=(4,7)$. The latter two allow us to ``concatenate" the two pairs in $\sempd{x\{a^*\}}$ and $\sempd{y\{b^*\}}$ to obtain a pair $((1,7),\mu)$, with $\mu(x)=(1,4)$, and $\mu(y)=(4,7)$. Note that this also implies that $\mu \in \semd{x\{a^*\}\cdot y\{b^*\}}$, since its corresponding span equals the entire document.

Notice that ``concatenating" $\mu_1$ and $\mu_2$ above is possible, since they share no variables. If we were dealing with an expression of the form $x\{a^*\}\cdot x\{b^*\}$, this would no longer be the case, and no mapping would be produced as the output of the expression. Some other ``pathological" cases such as $x\{x\{R\}\}$, which wants to bind $x$ inside itself, are also limited by our semantics, as this formula can never output any mappings.

On the other hand, some formulas that intuitively can make sense, but were not covered by the definition of functional regex in \cite{FKRV15}, have a clearly defined semantics in our setting. One such example would be the expression $e = (x\{(a\vee b)^*\} \vee y\{(a\vee b)^*\})^*$, which uses a Kleene star over a subexpression containing variables. If evaluated over the document $d$, this expression can output several mappings. For instance, we have that $((1,4),\mu_1)\in\sempd{y\{(a\vee b)^*\}}$, with $\mu_1(y)=(1,4)$ and that $((4,7),\mu_2)\in \sempd{x\{(a\vee b)^*\}}$, with $\mu_2(x)=(4,7)$. From this we can conclude that $\mu\in \semd{e}$, where $\mu = \mu_1 \cup \mu_2$.
\qed
\end{example}

It is worthwhile mentioning that the denotational semantics introduced here is much simpler than the semantics of variable regex defined in~\cite{FKRV15}. In Table~\ref{tab-semantics}, we give the semantics of our framework directly in terms of spans and mappings.
On the other hand, the semantics of variable regex in \cite{FKRV15} is given through the so-called \emph{parse trees}: syntactical structures that represent the evaluation of an expression over a document, while in \cite{F17}, the semantics uses reference words and projection functions.
We believe that one important contribution of our work is the simplification of the semantics by using mappings, which could help in the future to better understand variable regex and other IE languages.

Of course, there are ways to allow adding partial information in regex formulas without using mappings. For instance, one could simply map each variable that does not get assigned to the empty span $\varepsilon$. That is, an expression of the form $x\{R\}$ could be replaced with $x\{R\vee \varepsilon\}$, with $\varepsilon$ signifying that the variable is not assigned. One problem with this approach is that the term $\varepsilon$ already has a meaning in regex that is reserved for the empty word, which one would sometimes like to assign (e.g. to specify a landmark, or in the Kleene closure). Similarly, one could introduce a special \texttt{NULL} value to denote a variable that is not assigned, and add a \texttt{NULL} expression into regex formulas to signify that a subexpression was not matched to any span. The main drawback of this approach is that it would change the syntactic structure of regex, making them somewhat more cumbersome and less intuitive. On the other hand, none of these problems are present when we use mappings, as they both preserve the syntax of regex formulas, do not overwrite the previously defined semantics in border cases, and offer an elegant general definition encompassing other approaches and simplifying the definitions of \cite{FKRV15,F17}, while at the same time being fully declarative.

\subsection{Automata that extract information}\label{ss-automata}

In this subsection, we define automata models that support incomplete information extraction. Just as with \RGX, the definitions of the automata models come from \cite{FKRV15}, however, we need to redefine the semantics to support mappings. 

A \emph{variable-set automata} (\(\VA\)) is an automata model extended with captures variables in a way analogous to \RGX; that is, it behaves as a usual finite state automaton, except that it can also open and close variables. 
Formally, a $\VA$ automaton $A$ is a tuple \((Q, q_0, q_f, \delta)\),
where \(Q\) is a finite set of \textit{states}, $q_0$ and $q_f$ are the initial and the final state, respectively, and $\delta$
is a \textit{transition relation} consisting of letter transitions $(q, a, q')$, and variable transitions $(q, \Open{x}, q')$ or $(q,
\Close{x}\,,q')$, where $q, q' \in Q$, $a \in \Sigma$ and $x \in \VV$. 
The \(\vdash\) and \(\dashv\) are special symbols to denote the opening or closing of a variable $x$.
We define the set $\Var(A)$ as the set of all variables $x$ such that $\Open{x}$ appears in some transition of $A$.

\smallskip

\noindent{\bf Semantics.} A configuration of a $\VA$ automaton over a document $d$ is a tuple $(q, i)$ where $q \in Q$ is the current state and \(i \in [1, |d| + 1]\) is the \textit{current position} in $d$.
A run $\rho$ over a document $d = a_1 a_2 \cdots a_n$ is a sequence of the form:
$$
\rho \ = \ (q_0, i_0) \ \trans{o_1} \ (q_1, i_1) \ \trans{o_2} \ \cdots \ \trans{o_m} \ (q_m, i_m)
$$
where $o_j \in \Sigma \cup \{\Open{x}, \Close{x} \mid x \in \VV\}$, $(q_j, o_{j+1}, q_{j+1}) \in \delta$ and $i_0, \ldots, i_n$ is an increasing sequence such that $i_0 = 1$, $i_m = |d|+1$, and $i_{j+1} = i_j +1$ if $o_{j+1} \in \Sigma$ (i.e.\ the automata moves one position in the word only when reading a letter) and $i_{j+1} = i_j$ otherwise. Furthermore, $\rho$ must satisfy that variables are opened and closed in a correct manner, that is, each $x$ is opened or closed at most once and, if $x$ is closed at some position, then there must exists a previous position in $\rho$ where $x$ was opened. Note that we allowed $A$ to open $x$ without closing it, assuming that $x$ was never used in this case.
We say that $\rho$ is \emph{accepting} if $q_m = q_f$ in which case we define the mapping $\mu^{\rho}$ that maps $x$ into $(i_j, i_k) \in \sub(d)$ if, and only if, $o_{i_j} = \Open{x}$ and $o_{i_k} = \Close{x}$ in~$\rho$.
Finally, the semantics of $A$ over $d$, denoted by \(\semd{A}\)
is defined as the set of all $\mu^{\rho}$ where $\rho$ is an accepting run of $A$ over $D$.

Following \cite{FKRV15} we also redefine the semantics of the so-called  {\em variable-stack automata} (\(\VAStk\)), a restricted class of $\VA$ which only allow defining mappings that are hierarchical as in the case of \RGX. The new version of variable-stack automata is almost identical to the one of $\VA$ automata above, but we now restrict to runs $\rho$ where variables are open and closed following a stack policy. To avoid repeating the same definition we refer the reader to either \cite{FKRV15} or the appendix of this paper for more details. 
%
Lastly, we say that a {\em \VA\ is hierarchical} if every mapping it produces is hierarchical.

\subsection{Extracting information using rules}\label{ss-rules}
In \cite{AMRV16} a simplified version of RGX, called span regular expressions, was introduced. Formally, {\em span regular expressions}, or \SRGX\ for short, are \RGX\ formulas where all subexpressions of the form \(x\{\gamma\}\)  have $\gamma=\Sigma^*$. That is, in $\SRGX$, we have no control over the shape of the span we are capturing, and we cannot nest variables. For simplicity, we will often omit ${\Sigma^*}$ after variables when showing these formulas and simply write e.g. $a\cdot x \cdot a^*$ to denote the expression $a\cdot x\{\Sigma^*\} \cdot a^*$.

In order to allow specifying the shape of a span captured by some variable, \cite{AMRV16} allows joining \SRGX\ formulas using a rule-like syntax similar to Datalog. For instance to specify that the span captured by the variable $x$ in the expression above must conform to a regular expression $R$, we would write $a\cdot x \cdot a^* \wedge x.R$.

To define such rules formally, in our language 
we will allow two types of formulas: $R$ and $x.R$, where
$R$ is a \SRGX\ formula and $x$ a variable. The former is meant to be evaluated over the entire document,
while the latter applies to the span captured by the variable $x$. The semantics of the
extraction formula $R$  over a document $d$ is defined as in Table \ref{tab-semantics} above, and for $x.R$ as follows:
$$
  \semd{x.R} \ = \ \{ \mu \mid \exists s. (s,\mu)\in \sempd{x\{R\}} \}.
$$
We can now define rules for extracting information from a document as conjunctions
of extraction formulas. Formally,
 an \emph{extraction rule} is an expression of the form:
\begin{equation*} \label{eq-ann-rule}
\varphi = \varphi_0 \wedge x_1.\varphi_1 \wedge \cdots \wedge x_m.\varphi_m \tag{\dag}
\end{equation*}
where $m \geq 0$, all $\varphi_i$ are \SRGX\ formulas, and $x_i$ are variables\footnote{For simplicity we assume that there is only one formula applying to the entire document; namely $\varphi_0$. It is straightforward to extend the definitions below to include multiple formulas of this form.}. Extraction rules typically have an implication symbol and a \emph{head predicate}, which we will omit because it does not affect the analysis performed in this paper.

\smallskip

\noindent{\bf Semantics.} While \cite{AMRV16} has a simple definition of the semantics of extraction rules, lifting this definition to the domain of mappings requires us to account for nondeterminism of our expressions. What we mean by this is perhaps best captured by the rule $(x \vee y) \wedge x.(ab^*) \wedge y.(ba^*)$, where we first choose which variable is going to be mapped to the entire document, and then we need to satisfy its respective constraint. For instance, if $x$ is matched to the document, we want it to conform to the regular expression $ab^*$; however, in this case we do not really care about the content of $y$, so we should leave our mapping undefined on this variable.


Formally, we define when a rule of the form (\ref{eq-ann-rule}) is satisfied by a tuple of mappings $\overline{\mu}=(\mu_0,\mu_1,\ldots ,\mu_m)$. To avoid the problem mentioned above, we need the concept of {\em instantiated variables} in our tuple of mappings. 
For a rule $\varphi$ of the form (\ref{eq-ann-rule}) and a tuple of mappings $\overline{\mu}=(\mu_0,\mu_1,\ldots ,\mu_m)$ we define the set of \emph{instantiated variables}, denoted by $\ivar(\varphi,\overline{\mu})$ as the minimum set such that $dom(\mu_0)\subseteq \ivar(\varphi,\overline{\mu})$ and if $x_i\in \ivar(\varphi,\overline{\mu})$, then 
$dom(\mu_i)\subseteq \ivar(\varphi,\overline{\mu})$.
Intuitively, we want to put in $\ivar(\varphi,\overline{\mu})$ only the variables which are used in nondeterministic choices made by $\varphi$ and $\overline{\mu}$. For instance, in the rule $(x \vee y) \wedge x.(ab^*) \wedge y.(ba^*)$, if we decide that $x$ should be matched to our document, then we will not assign a value to the variable $y$ and vice versa. We now define that a tuple of mappings $\overline{\mu}=(\mu_0,\mu_1,\ldots ,\mu_m)$ satisfies $\varphi$ over a document $d$, denoted by $\overline{\mu}\models_d \varphi$, if the following three conditions hold: (1) $\mu_0\in \semd{\varphi_0}$; (2) $\mu_i \in \semd{x_i.\varphi_i}$  whenever $x_i\in \ivar(\varphi,\overline{\mu})$ and $\mu_i = \emptyset$ otherwise; and (3) $\mu_i\sim \mu_j$ for all $i,j$.
Here the last condition will allow us to ``join'' all the mappings capturing each subformula $\varphi_i$ into one. The problem with nondeterminism is handled by condition (2), since we force all instantiated variables to take a value, and the non-instantiated ones to be undefined. Finally, condition (1) starts from $\varphi_0$ which refers to the entire document and serves as a ``root'' for our mappings.

We can now define the semantics of an extraction rule $\varphi$ over a document
$d$ as follows:
$$
  \semd{\varphi} \ = \ \{ \mu \mid \exists \overline{\mu} \text{ such that } \overline{\mu}\models_d \varphi \text{ and } \mu = \bigcup_i \mu_i \},
$$
where $\bigcup_i \mu_i$ denotes the mapping defined as the union of all $\mu_i$.

\section{Expressiveness of IE languages}
\label{sec:comp}

In this section we compare how different IE approaches compare in terms of expressive power. We first show how the new semantics based on mappings subsumes the relation based semantics of $\RGX$ from \cite{FKRV15} and $\SRGX$ from \cite{AMRV16}. Next, we show that the results of \cite{FKRV15} comparing automata models from Section \ref{ss-automata} and regex formulas can be lifted to support incomplete information. We finish with a comparison of the rule-based language introduced in Section \ref{ss-rules} with $\RGX$.

\subsection{Mapping-based semantics and relation-based semantics}\label{ss-comparison}

Having the general definition of formulas which define mappings, we can now show how this framework subsumes regex formulas as defined in \cite{FKRV15} and span regular expressions from \cite{AMRV16}.

We start with regex formulas of \cite{FKRV15}. Although the expressions from \cite{FKRV15} use the same syntax as our \RGX\ formulas, the setting of \cite{FKRV15} dictates that document spanners always define relations. This automatically excludes expressions such as $R_1\vee R_2$ from Section~\ref{ss-regex} which allows mappings with different domains. What \cite{FKRV15} proposes instead is that each mapping defined by an expression assigns precisely the same variables every time (and also all of them); that is, we want our expressions to act as functions. As shown in \cite{FKRV15} there is a very easy syntactic criteria for this, resulting in functional \RGX\ formulas.

A \(\RGX\) \(\gamma\) is called \textit{functional with respect to the set of variables} $X$ (abbreviated as functional wrt $X$) if one of the following syntactic restrictions holds:
\begin{itemize}\itemsep=0pt
  \item \(\gamma \in \Sigma \cup \{\varepsilon\}\) and $X=\emptyset$.
  \item \(\gamma = \varphi_1 \vee \varphi_2\), where \(\varphi_1, \varphi_2\) are
    \emph{functional} wrt $X$. 
  \item \(\gamma = \varphi_1 \cdot \varphi_2\), where \(\varphi_1\) is functional wrt $X'\subseteq X$ and \(\varphi_2\)
    is functional wrt $X\backslash X'$. 
  \item \(\gamma = {(\varphi)}^{*}\), where \(\Var(\varphi) = \emptyset\) and $X=\emptyset$.
  \item \(\gamma = x\{\gamma'\}\) where $x\in X$ and $\gamma'$ is functional with respect to $X\backslash\{x\}$.  
\end{itemize}
A \(\RGX\) \(\gamma\) is called \textit{functional} if it is functional with respect to $\Var(\gamma)$.

This condition ensures that each variable mentioned in \(\gamma\) will appear
exactly once in every word that can be derived from \(\gamma\), when we treat
\(\gamma\) as a classical regular expression with variables as part of the
alphabet. We refer to the class of \emph{functional} \(\RGX\)s as \(\fRGX\).
Note that this corresponds to the original definition of regex formulas given by
\cite{FKRV15}, even when we consider the new semantics. Thus, we have:

\begin{theorem}\label{th-frgx}
Regex formulas of \cite{FKRV15} are equivalent to the class \(\fRGX\) defined above.
\end{theorem}


Next, we show how \RGX\ formulas subsume span regular expressions of \cite{AMRV16}. For this, observe that span regular expressions of \cite{AMRV16}
have the same syntax as $\SRGX$; that is, they can be seen as \RGX\ formulas where all subexpressions of the form \(x\{\gamma\}\) have $\gamma=\Sigma^*$. 

To compare \SRGX\ with span regular expressions, we also need to take note of the semantics proposed in \cite{AMRV16}. One problem with that semantics is that when a variable is not matched by the expression, the resulting mapping is assigned an arbitrary span, which can be rather misleading (e.g.\ in the sales example above we could not determine if the tax data is real or assigned arbitrarily). Of course, this type of behaviour can easily be simulated by ``joining'' the results obtained by a \SRGX\ with the set of all total mappings. Another, more subtle problem, is that the formalism of \cite{AMRV16} allows expressions of the form $x\{\Sigma^*\}\cdot x\{\Sigma^*\}$ (forcing $x$ to be assigned the empty string at the same position multiple times), while this \RGX\ is not satisfiable. We call span regular expressions which prohibit such behaviour \emph{proper}. We now obtain the following.

\begin{theorem}\label{theo:equiv-SRE}
  Let \(d\) be a document, \(\gamma\) be a \(\RGX\), \(M\) be the set of all
  total functions from \(\Var(\gamma)\) to \(\sub(d)\), and let \(\semd{\gamma}' = M
  \Join \semd{\gamma}\). Under these semantics, \(\SRGX\) and proper span regular expressions of \cite{AMRV16} are
  equivalent.
\end{theorem}
We can therefore conclude that using mappings is indeed a natural extension of the previous semantics of $\RGX$ and $\SRGX$.


\subsection{Comparing expressions to automata}\label{ss-previous}

One of the main problems studied in \cite{FKRV15} was to determine the relationship between the automata models from Section \ref{ss-automata} (restricted to always output relations) with the class of functional \RGX\ formulas. As our framework is an extension (in terms of expressiveness) and a simplification (in terms of semantics) of \cite{FKRV15} that allows mappings instead of simple relations, here we show how the main results on $\VA$ and $\fRGX$ from \cite{FKRV15} can be generalised to our setting. 
We start by showing that the class of \RGX\ formulas is also captured by $\VAStk$ automata in our new setting.

\begin{theorem}[\cite{FKRV15}]\label{th-stack_rgx}
Every \VAStk\ automaton has an equivalent \RGX\ formula and vice versa. That is $\VAStk \equiv \RGX$.
\end{theorem}

Just as in the proof for the relational case \cite{FKRV15}, the main step is to show that $\VAStk$ automata can be simplified by decomposing them into an (exponential) union of disjoint paths known as $\PUStk$ (path union $\VAStk$). In $\PUStk$ automata each path is essentially a functional \RGX\ formula, thus making the transformation straightforward. The only difference to the proof of \cite{FKRV15} is that when transforming $\VAStk$ automaton into a union of paths, we need to consider all paths of length at most $2\cdot k+1$ in order to accommodate partial mappings, where $k$ is the number of variables. The notion of a consistent path also changes, since we are allowed to open a variable, but never close it. As a corollary we get that every \RGX\ is equivalent to a (potentially exponential) union of functional \RGX\ formulas (with this union being empty when the \RGX\ is not satisfiable).

Similarly as in the functional case, it is also straightforward to prove that the mappings defined by $\VAStk$ and \RGX\ are hierarchical. Furthermore, just as in \cite{FKRV15}, one can show that the class of $\VA$ automata which produce only hierarchical mappings is equivalent to \RGX\ in the general case.

\begin{theorem}[\cite{FKRV15}]\label{th-rgx_equiv_hierarchical_va}
Every \VA\ automaton that is hierarchical has an equivalent \RGX\ formula and vice versa.
\end{theorem}

Both $\VA$ and \VAStk\ automata, as well as \RGX, provide a simple way of extracting information. To permit a more complex way of defining extracted relations, \cite{FKRV15} allows combining them using basic algebraic operations of union, projection and join. While defining a union or projection of two automata or \RGX\ is straightforward, in the case of join we now use joins of mappings instead of the natural join (as used in \cite{FKRV15}). Formally, for two \VA\ automatons $A_1$ and $A_2$, we define the ``join automaton" $A_1\Join A_2$ using the following semantics: for a document $d$, we have $\semd{A_1\Join A_2} = \semd{A_1}\Join \semd{A_2}$. We denote the class of extraction expressions obtained by closing \VA\ under union, projection and join with $\VA^{\{\cup,\pi,\Join\}}$, and similarly for \VAStk\ and \RGX.

To establish a relationship between algebras based on $\VAStk$ and $\VA$ automata, \cite{FKRV15} shows that $\VA$ is closed under union, projection and join. We can show that the same holds true when dealing with mappings, but now the proofs change quite a bit. That is, while closure under  projection is much  easier to prove in our setting, closure under join now requires an exponential blowup, since to join mappings, we need to keep track of variables opened by each mapping in our automaton. Similarly, \cite{FKRV15} shows that each \VA\ automaton can be expressed using the expressions in the algebra $\VAStk^{\{\cup,\pi,\Join\}}$; as this proof holds verbatim in the case of mappings we obtain the following.

\begin{theorem}[\cite{FKRV15}]\label{th-va_algebra_equiv}
$\VA^{\{\cup,\pi,\Join\}}\equiv\VA\equiv\VAStk^{\{\cup,\pi,\Join\}}.$
\end{theorem}

As we showed here, the main results from \cite{FKRV15} can be lifted to hold in the more general setting of mappings, thus suggesting that the added generality does not impact the intuition behind the extraction process.

\subsection{Comparing $\RGX$ with rules}\label{sec:expr}

In this subsection we will compare the expressive power of two different frameworks for extracting information: $\RGX$ formulas of \cite{FKRV15} and extraction rules of \cite{AMRV16}. We do this under the new semantics allowing incomplete information and show that, while in general the two languages are not comparable, by simplifying extraction rules we can capture $\RGX$.

Extraction rules allow us to define complex conditions about the spans we wish to extract. For instance, if we wanted to extract all spans whose content is a word belonging to (ordinary) regular expressions $R_1$ and $R_2$ at the same time, we could use the rule $\Sigma^*\cdot x \cdot \Sigma^* \wedge x.R_1 \wedge x.R_2.$ More importantly, using extraction rules, we can now define valuations which cannot be defined using \RGX, since they can define mappings which are not hierarchical. For instance, the rule $x\wedge x.ayaa \wedge x.aaza$ is one such rule, since it makes $y$ and $z$ overlap on the document $aaaaa$.
In some sense, the ability of rules to use conjunctions of variables makes them
more powerful than \RGX\ formulas. 
On the other hand, the ability of \RGX\ formulas to use
disjunction of variables poses similar problems for \SRGX. 

\begin{theorem} \label{prop-rgx-rules}
Extraction rules and \RGX\ are incomparable in terms of the expressive power.
\end{theorem}
In light of this result, we study how the class of extraction rules can be pruned in order to capture \RGX.

\medskip

\noindent{\bf Simplifying extraction rules.}  As discussed above, the capability of an extraction rule to use conjunctions of the same variable multiple times already takes them outside of the reach of \RGX. Therefore, the most general class of rules we will consider disallows that type of behaviour. We call such rules {\em simple rules}. Formally, an extraction rule $\varphi$ of the form (\ref{eq-ann-rule}) is {\em simple}, if all $x_i$ are pairwise distinct. 
From now on, we assume that all classes of rules considered in this section are simple. 

Another feature that makes rules different
from \RGX\ is their ability to enforce cyclic behaviour through expressions of the form $x.y \wedge y.ax$.
A natural way to circumvent this shortcoming is to force the rules to have an acyclic structure. 
In fact, this kind of restriction was already considered in \cite{AMRV16}, as it allows faster evaluation than general rules.
Therefore, a natural question at this point is if the capability of rules to define cycles is really useful, or if they can be removed.
We answer now the question whether cycles can be eliminated from rules, and somewhat surprisingly show that, while generally not possible, in the case of rules defined by functional \SRGX\ this is indeed true.

In order to study the cyclic behaviour of rules, we first need to explain how each rule can be viewed as a graph.
To each extraction rule $\varphi=\varphi_0 \wedge x_1.\varphi_1 \wedge \cdots \wedge x_m.\varphi_m$ we associate a graph $G_\varphi$ defined as follows. The set of nodes of $G_{\varphi}$ contains all the variables $x_1,\ldots ,x_m$ plus one special node labelled $\texttt{doc}$ corresponding
to the formula $\varphi_0$.
There exists an edge $(x, y)$ between two variables in $G_{\varphi}$ if, and only if, there is an
\emph{extraction formula} $x.R$ in $\varphi$ such that $y$ occurs in
$R$. Furthermore, if the variable $x$ occurs in the formula $\varphi_0$, we add an edge $(\texttt{doc},x)$ to $G_\varphi$.
Then we say that a simple rule $\varphi$ is {\em dag-like}, if the graph $G_\varphi$ contains no cycles, and \emph{tree-like} if $G_{\varphi}$ is a tree rooted at $\texttt{doc}$. 

To answer the question whether cycles can be eliminated from rules, let us consider most general case; namely, simple rules over full \SRGX. It is straightforward to see that in a rule of the form $(x \vee y) \wedge x.(y \vee \Sigma^*) \wedge y.(x \vee \Sigma^*)$, the cycle formed by $x$ and $y$ cannot be broken and the rule cannot be rewritten as a single dag-like rule. The main obstacle here is the fact that in each part of the rule we make a nondeterministic choice which can then affect the value of all the variables.
However, there is one important class of expressions, which would prohibit our rules to define properties such as the one above; that is, functional \SRGX (we call a $\SRGX$ functional if the underlying $\RGX$ is functional). 
In the next result, we show that in the case of functional rules (i.e.\ rules defined by functional \SRGX) cycles can always be removed, and in fact, converting a simple functional rule into a dag-like rule takes only polynomial time.

\begin{theorem}\label{thm:simple_to_dag}
	For every simple rule that is functional there is an equivalent (functional) dag-like rule. Moreover, we can obtain the equivalent rule in polynomial time.
\end{theorem}

It is remarkable that the algorithm for removing cycles runs in polynomial time and, furthermore, it produces a single rule. 
We think that this result is interesting in its own right and potentially useful in other contexts regarding the use of rules in information extraction.

\medskip

\noindent{\bf Unions of simple rules capture RGX.}
We now know that cycles can be eliminated from functional rules, but is there any way to removes cycles from rules that are non-functional? Moreover, can we go even further from dag-like rules, and convert each rule into a tree-like rule? 
Unfortunately, one can easy show that all these questions have a negative answer since non-functional cyclic rules, and even functional dag-like rules, have the ability to express some sort of disjunction.
For this reason, we introduce here the class of \emph{unions of simple rules} and compare its expressive power with \RGX.  
Formally, {\em union of simple rules} is a set of simple rules $A$. The semantics $\semd{A}$ over a 
document $d$ is defined as all mapping $\mu$ over $d$ such that $\mu \in \semd{\varphi}$ for some $\varphi\in A$.

We start by extending our results for removing cycles from functional to non-functional rules. 
As it turns out, although functional and non-functional rules are not equivalent, every non-functional simple rule can in fact be expressed as a union of functional rules.
Then, by combining this fact with Theorem~\ref{thm:simple_to_dag} one can show that each non-functional rule can be made acyclic by transforming it to a  \emph{union of dag-like rules}. 

\begin{proposition}\label{prop:nonfunc}
	Every simple rule is equivalent to a union of functional dag-like rules. 
\end{proposition}

Now that the connection with union of acyclic rules is settled, our next step is to understand when dag-like rules can be defined by \RGX\ formulas and, moreover, when can they be converted into tree-like rules.
First, observe that a functional \RGX\ formula is always satisfiable; namely, there is always a document on which there is an assignment satisfying this formula. Similarly, every functional tree-like rule is also satisfiable. On the other hand, the functional simple rule $x\wedge x.y \wedge y.ax$ is clearly not satisfiable, since it forces $x$ and $y$ to be equal and different at the same time. Therefore, to link rules with \RGX, we should consider only the satisfiable ones.

\begin{proposition}\label{thm:dag_to_tree}
	Every dag-like rule that is satisfiable is equivalent to a union of functional tree-like rules.
\end{proposition}

The idea of the proof here is similar to the cycle elimination procedure of Theorem \ref{thm:simple_to_dag}, but this time considering undirected cycles. One can show that eliminating undirected cycles results in a double exponential number of tree-like rules. In case that the rule was not satisfiable, our algorithm will simply abort.


With this at hand, we can now describe the relationship between unions of simple rules and \RGX. Indeed, a union of simple rules is equivalent to a  union of dag-like rules by Proposition~\ref{prop:nonfunc} and this union is equivalent to a union of functional tree-like rules by Proposition~\ref{thm:dag_to_tree} (if some dag-like rule is not satisfiable, we just output an unsatisfiable non-fuctional RGX formula in our algorithm from Proposition \ref{thm:dag_to_tree}). 
Then one can easily see that any functional tree-like rule $\varphi$ is equivalent to a $\RGX$ formula given that each (singleton) formula $x.R$ in $\varphi$ can be removed by composing the tree structure recursively with formulas of the form $x\{R\}$.
Conversely, one can show that each $\RGX$ formula can be defined as a union of simple rules.

\begin{theorem}\label{thm:rgx:simple}
	\RGX\ formulas and unions of simple rules are equivalent. Moreover, every \RGX\
  formula is equivalent to a union of tree-like rules.
\end{theorem}


\section{Evaluation of languages for extracting incomplete data}\label{sec:evaluation}

In this section, we study the computational complexity of evaluating an extraction expression $\gamma$ over a document $d$, namely, the complexity of enumerating all mappings $\mu \in \semd{\gamma}$.
Given that we are dealing with an enumeration problem, our objective is to obtain a \emph{polynomial delay} algorithm~\cite{johnson1988generating}, i.e., an algorithm that enumerates all the mappings in $\semd{\gamma}$ by taking  time polynomial in the size of $\gamma$ and $d$ between outputting two consecutive results. For this analysis, our objective is to determine which decision problems can be used to faithfully model the process of enumerating all the outputs of an IE expression, and then study their complexity. We formally define our decision problems in Subsection \ref{ss:enum} and show that in full generality, none of the languages we consider can be enumerated efficiently. In Subsection \ref{ss:tractable} we then identify several fragments that can be evaluated with a polynomial delay.

\subsection{Decision problems for enumeration}\label{ss:enum}
In order to formally define the decision problems modelling query enumeration we need to introduce some notation first. 
Let $\bot$ be a new symbol. 
An \emph{extended mapping} $\mu$ over $d$ is a partial function from $\VV$ to $\sub(d) \cup \{\bot\}$.
Intuitively, in our decision problem $\mu(x) = \bot$ will represent that the variable $x$ will not be mapped to any span.
Furthermore, we usually treat $\mu$ as a normal mapping by assuming that $x$ is not in $\dom(\mu)$ for all variables $x$ that are mapped to $\bot$.
Given two extended mappings $\mu$ and $\mu'$, we say that $\mu \subseteq \mu'$ if, and only if, $\mu(x) = \mu'(x)$ for every $x \in \dom(\mu)$. 
Then for any language $\LL$ for information extraction we define the main decision problem for evaluating expressions from $\LL$, called $\EVAL[\LL]$, as follows:
\begin{center}
	\framebox{
		\begin{tabular}{rl}
			\textbf{Problem:} & $\EVAL[\LL]$\\
			\textbf{Input:} &  An expression $\gamma \in \LL$, a document $d$,  \\
			& and an extended mapping $\mu$.  \\
			\textbf{Question:} & Does there exist $\mu'$ such that \\
			&  $\mu \subseteq \mu'$ and  $\mu' \in \semd{\gamma}$?
		\end{tabular}
	}
\end{center}
In other words, in $\EVAL[\LL]$ we want to check whether $\mu$ can be extended to a mapping $\mu'$ that satisfies $\gamma$ in $d$.
Note that in our analysis we will consider the combined complexity of $\EVAL[\LL]$.

We claim that $\EVAL[\LL]$ correctly models the problem of enumerating all mappings in $\semd{\gamma}$. 
Indeed, if we can find a polynomial time algorithm for deciding $\EVAL[\LL]$, one can have a polynomial delay algorithm for enumerating the mappings in $\semd{\gamma}$ as given in Algorithm \ref{alg:enum1}.

\begin{algorithm}
  \caption{Enumerate all spans in \(\semd{\gamma}\)}
  \label{alg:enum1}
  \begin{algorithmic}[1]
    \Procedure{Enumerate}{$\gamma, d, \mu, V$}
      \If{\(V = \emptyset\)}
        \State{\textbf{output} \(\mu\) and \Return}
      \EndIf
      \State{Let \(x\) be some element from \(V\)}
      \For{\(s \in \sub(d) \cup \{\bot\}\)}
        \If{\(\EVAL[\LL](\gamma, d, \mu[x \to s])\)}
          \State{
            \Call{Enumerate}{$\gamma, d, \mu[x \to s], V\,\setminus\,\{x\}$}}
        \EndIf
      \EndFor
    \EndProcedure
  \end{algorithmic}
\end{algorithm}

The procedure starts with the empty mapping $\mu = \emptyset$ and the set $V$ of variables yet to be assigned equal to $\Var(\gamma)$. 
For a variable $x \notin \dom(\mu)$ we iterate over all $s \in \sub(d)$ (or the symbol $\bot$ signalling that $x$ is not assigned) and check if $\EVAL[\LL](\gamma, d, \mu[x \rightarrow s])$ is true where $\mu[x \rightarrow s]$ is an extended mapping where $x$ is assigned to $s$ (lines 4 through 6).
If the answer is positive, then in line 7 we recursively continue with the mapping $\mu[x \rightarrow s]$ (i.e. we know that the set of answers is non-empty).
Finally, we print the mapping $\mu$ when all variables in $\Var(\gamma)$ are assigned a span or the symbol $\bot$ (i.e. $V=\emptyset$ in line 2). 

We can therefore obtain the following.
\begin{theorem}\label{prop:evaluation}
	If $\EVAL[\LL]$ is in $\PTIME$, then enumerating all mappings in $\semd{\gamma}$ can be done with polynomial delay.
\end{theorem}

Notice that Theorem \ref{prop:evaluation} is a general result allowing us to reason about efficient enumeration of IE languages. That is, when we want to show that {\em any} IE language $\LL$ can be enumerated efficiently, we simply need to show that $\EVAL[\LL]$ is in $\PTIME$. This is in contrast with approaches such as \cite{FKP17}, which, while providing a faster algorithm than the ones we derive below, are applicable to a single fixed language $\LL$. 

Before continuing we would like to stress the importance of selecting the correct decision problem to model query enumeration. Indeed, while $\EVAL[\LL]$ might seem somewhat counter intuitive at a first glance, as Theorem \ref{prop:evaluation} shows, efficiently solving $\EVAL[\LL]$ gives an efficient enumeration procedure. A more common variation of the evaluation problem, would ask if, given a mapping $\mu$, an expression $\gamma\in \LL$, and a document $d$, it holds that $\mu \in \semd{\gamma}$. We call this version of evaluation model checking and denote it with $\mdcheck[\LL]$. Model checking problem for subclasses of variable set automata that output relations was studied in \cite{F17} (under the name evaluation), where a \PTIME\ algorithm is given for a subclass of VA automata. Unfortunately, solving model checking efficiently does not help us with the enumeration problem, since we would have to check each mapping one by one -- a task that can produce an exponential gap between two consecutive outputs. On the other hand, it is straightforward to see that model checking is a special case of $\EVAL$.

%
%

Notice, however, that showing $\EVAL[\LL]$ to be hard does not necessarily rule out the existence of a polynomial delay enumeration procedure for $\LL$. For this, we need to consider a related problem of checking non-emptiness.
Formally, the {\em non-emptiness problem}, denoted $\nemp[\LL]$, asks, given a document $d$ and an expression $\gamma$, whether $\semd{\gamma}= \emptyset$. One can easily see that non-emptiness is actually a restricted instance of $\EVAL[\LL]$, namely: $\nemp[\LL](\gamma, d) =  \EVAL[\LL](\gamma, d, \emptyset)$.
This implies that if we find an efficient algorithm for $\EVAL[\LL]$ then the same holds for $\nemp[\LL]$, and that showing $\nemp[\LL]$ to be $\NP$-hard implies the same for $\EVAL[\LL]$.
More importantly, if we can show that $\nemp[\LL]$ is difficult, then no polynomial delay algorithm for $\LL$ can exist (under standard complexity assumptions), as we could simply run the enumeration procedure until the first output is produced. Note on the other hand that showing e.g. $\NP$-hardness of $\mdcheck[\LL]$ does also not necessarily imply that efficient enumeration is not possible. 
As we are interested in query enumeration, we will therefore not consider the model checking problem in the remainder of this paper. 

We would like to note that \cite{AMRV16} and \cite{F17} already considered the non-emptiness problem and the model checking problem. In the following results we will point out when a (weaker) version of our result was proved in one of the two works. Generally, we can use \cite{AMRV16,F17} to derive some lower bounds, while we need to show the matching upper bound (when possible) separately. It is important to stress that what \cite{F17} calls evaluation is our model checking problem, and, as discussed above, can not be used to obtain an efficient algorithm for enumeration or $\EVAL$, nor tell us when enumeration with polynomial delay is not possible.


Now that we identified the appropriate decision problem, we start by understanding the complexity of $\EVAL[\LL]$ in the most general case.
It is easy to see that checking $\EVAL[\LL]$ is in $\NP$ for all languages and computational models considered in this paper. Indeed, given a mapping $\mu'$ such that $\mu \subseteq \mu'$ one can check in $\PTIME$ if $\mu' \in \semd{\gamma}$ by using finite automata evaluation techniques~\cite{HopcroftU79}. As the following result shows, this is the best that one can do if $\RGX$ or variable-set automata contain the language of $\SRGX$, as non emptiness is already hard for this fragment.
\begin{theorem}\label{theo:eval-hardness-general}
$\nemp[\SRGX]$ is $\NP$-compl.
\end{theorem}
We would like to remark that this result was proved in \cite{AMRV16} and here we strengthen it to allow using partial mappings.

\subsection{Tractable fragments}\label{ss:tractable}
Since Theorem \ref{theo:eval-hardness-general} implies that efficiently enumerating answers of $\RGX$ or variable-set automata is not possible unless $\PTIME=\NP$,
we now examine several syntactic restrictions that make their evaluation problem tractable.
Note that the previous negative results are considering a more general setting than the one presented in~\cite{FKRV15}, where $\RGX$ and variable-set automata are restricted to be \emph{functional} which forces them to only generate relations of spans. 
Interestingly, the functional restriction decreases the complexity of the evaluation problem for $\RGX$ as the following result shows.
\begin{proposition}\label{theo:ptime-functional-RGX}
	$\EVAL[\fRGX]$ is in $\PTIME$.
\end{proposition}
This result proves that the functional restriction for $\RGX$ introduced in~\cite{FKRV15} is crucial for getting tractability. 
The question that now remains is what the necessary restrictions are that make the evaluation of $\RGX$ tractable when outputting mappings and how to extend these restrictions to other classes like variable-set automata.
One possible approach is to consider variable-set automata that produce only relations. 
Formally, we say that a variable-set automaton $\cA$ is \emph{relational} if for all documents $d$, the set $\semd{\cA}$ forms a relation. 
As the next result shows, this semantic restriction is not enough to ensure the tractability of query enumeration.
\begin{proposition}\label{theo:eval-hardness-relational-VA}
	$\nemp$ of relational VA automata is $\NP$-complete\footnote{Here $\NP$-hardness of $\nemp$ already follows from the results of \cite{AMRV16,F17}.}.
\end{proposition}
By taking a close look at the proof of the previous result, one can note that a necessary property for getting intractability is that, during a run, the automaton can see the same variable on potential transitions many times but not use it if it has closed the same variable in the past.
Intuitively, this cannot happen in functional $\RGX$ formulas where for every subformula of the form $\varphi_1 \cdot \varphi_2$ it holds that $\Var(\varphi_1) \cap \Var(\varphi_2) =
\emptyset$.
Actually, we claim that this is the restriction that implies tractability for evaluating $\RGX$ formulas. 
Formally, we say that a $\RGX$ formula $\gamma$ is \emph{sequential} if for every subformula of the form $\varphi_1 \cdot \varphi_2$ or $\varphi^*$ it holds that $\Var(\varphi_1) \cap \Var(\varphi_2) =
\emptyset$ and $\Var(\varphi) = \emptyset$, respectively.
We can also extend these ideas of sequentiality from $\RGX$ formulas to variable-set automata as follows.
A path $\pi$ of a variable-set automaton $\cA = (Q, q_0, q_f, \delta)$ is a finite sequence of transitions 
$\pi: (q_1, s_2, q_2), (q_2, s_3, q_3) \ldots, (q_{m-1}, s_m, q_m)$ such that $(q_i, s_{i+1}, q_{i+1}) \in
\delta$ for all $i \in [1, m - 1]$.
We say that a path $\pi$ of $\cA$ is sequential if for every variable \(x \in
\VV\) it holds that: (1) there is at most one \(i \in [1, m]\) such that \(s_i = \Open{x}\); (2) if such an $i$ exists, then there is precisely one \(j \in [1, m]\) such that \(s_j = \Close{x}\); and (3) $i<j$.
We say that variable-set automaton $\cA$ is sequential if every path in $\cA$ is sequential.
Finally, we denote the class of sequential RGX and sequential variable-set automata by $\seqRGX$ and $\seqVA$, respectively.

The first natural question about sequentiality is whether this property can be checked efficiently. As the next proposition shows, this is indeed the case.
\begin{proposition}\label{theo:ptime-sequential-check}
	Deciding if an VA automaton is sequential can be done in $\NLOGSPACE$.
\end{proposition}

Sequentiality is a mild  restriction over extraction expressions since it still allows many $\RGX$ formulas that are useful in practice.
For example, all extraction expressions discussed in Section~\ref{sec-unifying} are sequential.
Furthermore, as we now show, no expressive power is lost when restricting to sequential \RGX\ or automata. 
\begin{proposition}\label{prop:rgx-equiv-rgx-seq}
For every \RGX\ (\VA\ automaton), there exists a sequential \RGX\ (sequential \VA, respectively) that defines the same extraction function.
\end{proposition}
We believe that sequentiality is a natural syntactical restriction\footnote{Note that \cite{F17} already considers a less general version of sequentiality  called functional VA automata, that open and close {\em all} the variables exactly once. There a version of Proposition \ref{theo:ptime-sequential-check} is given with a $\PTIME$ bound, as well as a version of Proposition \ref{prop:rgx-equiv-rgx-seq} for functional automata.} of how to use variables in extraction expressions. Namely, one should not reuse variables by concatenation since this can easily make the formula unsatisfiable.
Furthermore, the more important advantage for users is that $\RGX$ and VA automata that are sequential can be evaluated efficiently. 
\begin{theorem}\label{theo:ptime-sequential-sVA}
	$\EVAL[\seqRGX]$ and $\EVAL[\seqVA]$ is in $\PTIME$.
\end{theorem}
It is important to recall that this result implies, by Proposition~\ref{prop:evaluation}, that the evaluation of sequential $\RGX$ formulas can be done with polynomial delay. An interesting question we would like to tackle in the future is if this algorithm can be further optimised to yield a \emph{constant delay algorithm}~\cite{johnson1988generating} like the one presented in~\cite{AMRV16} for the so-called navigational formulas -- a class strictly subsumed by sequential $\RGX$.


Now that we have captured an efficient fragment of \(\RGX\), we will analyse what happens with the complexity of the evaluation problem for extraction rules. First, we show that evaluating rules is in general a hard problem. In fact, non-emptiness is already $\NP$-hard, even when restricted to dag-like rules with functional \(\SRGX\). 
\begin{theorem}\label{theo:eval-daglike-npcomplete}
  \(\nemp\) of functional dag-like rules is \(\NP\)-complete.
\end{theorem}
The difficulty in this case arises from the fact that dag-like rules allow referencing the same variable from different extraction expressions. 
%
A natural way to circumvent this is to use tree-like rules. 
Indeed, the fact that, in a tree-like rule, different branches are independent, causes the evaluation problem to become tractable.
In fact, the functionality constraint is not really needed here, as the result holds even for sequential rules.
\begin{theorem}\label{theo:treelike-ptime}
  \(\EVAL\) of sequential tree-like rules is in \(\PTIME\).
\end{theorem}

This implies that we should focus on sequential tree-like rules if we wish to have efficient algorithms for rules. Luckily, these do not come at a high price in terms of expressiveness, since Propositions~\ref{prop:nonfunc} and~\ref{thm:dag_to_tree} imply that every satisfiable simple rule is equivalent to a union of sequential tree-like rules.

The previous results show how far we can go when syntactically restricting the class of $\RGX$ formulas, variable-set automata, or extraction rules in order to get tractability.
The next step is to parametrise the size of the query not only in terms of the length, but also in terms of meaningful parameters that are usually small in practice.
In this direction, a natural parameter is the number of variables of a formula or automata since one would expect that this number will not be huge. 
Indeed, if we restrict the number of variables of a $\RGX$ formula or VA automata we can show that the problem is fixed parameter tractable.
\begin{theorem}\label{theo:RGX-FPT}
  $\EVAL[\RGX]$ and $\EVAL[\VA]$ parametrised by the number of variables is $\FPT$.
\end{theorem}



\section{Static analysis and complexity}\label{sec:eval}
\label{sec:complexity}

In this section, we study the computational complexity of static analysis problems for document spanners like satisfiability and containment.
Determining the exact complexity of these problems is crucial for query optimisation~\cite{AHV95} and data integration~\cite{lenzerini2002data}, and it gives us
a better understanding of how difficult it is to manage RGX formulas and $\VA$ automata.
%
We start with the satisfiability problem for RGX formulas and~\VA. 
Formally, let $\LL$ be any formalism for defining document spanners. Then the satisfiability problem of $\LL$, denoted  $\sat[\LL]$, asks given an expression $\gamma \in \LL$ if there exist  a document $d$ such that $\semd{\gamma}$ is non-empty.


$\sat[\LL]$ is a natural generalisation of the satisfiability problem for ordinary regular languages: if $\gamma$ does not contain variables, then asking if $\semd{\gamma} \neq \emptyset$ for some document $d$ is the same as asking if the language of $\gamma$ is non-empty.
It is a folklore result that satisfiability of regular languages given by regular expressions or NFAs has low-complexity~\cite{HopcroftU79}. 
Unfortunately, in the information extraction context, this problem is intractable even for $\SRGX$.
\begin{theorem}\label{theo:sat-general}
$\sat[\VA]$ and $\sat$ of extraction rules are $\NP$-complete. Furthermore, $\sat[\SRGX]$ is already $\NP$-hard. 
\end{theorem}
These results show that satisfiability is generally $\NP$-complete for all information extraction languages we consider in this paper.
The next step is to consider syntactic restrictions of RGX or $\VA$, like e.g.\ sequentiality introduced in Section~\ref{sec:evaluation}.
Indeed, with sequentiality we can restore tractability.
\begin{theorem}\label{theo:sat-sequential}
$\sat[\seqVA]$ is in $\NLOGSPACE$. 
\end{theorem}
It is interesting to note that this result is very similar to satisfiability of finite state automata: given a sequential $\VA$ the $\NLOGSPACE$ algorithm simply checks reachability between initial and final states.
This again shows the similarity between finite state automata and $\VA$ if the sequential restriction is imposed. 

Next, we consider extraction rules combined with the sequential or functional $\SRGX$.
Similarly as before, $\sat$ of extraction rules remains intractable even for the class of functional dag-like rules.
However, if we consider sequential tree-like rules we can restore tractability since tree-like rules are always satisfiable.
\begin{theorem}\label{theo:sat-rules}
	$\sat$ of functional dag-like rules is $\NP$-hard.
	On the other hand, any sequential tree-like rule is always satisfiable.
\end{theorem}
It is important to make the connection here between regular expressions, sequential RGX and sequential tree-like rules: all formalisms are trivially satisfiable.
In some sense, this gives more evidence that sequential RGX and sequential tree-like rules are the natural extensions of regular expressions, as they inherit all the good properties of its predecessor.

We continue by considering the classical problem of containment of expressions.
Formally, for a language $\LL$ we define the problem $\containment[\LL]$, which, given two expressions $\gamma_1$ and $\gamma_2$ in $\LL$, asks whether $\semd{\gamma_1} \subseteq \semd{\gamma_2}$ holds for every document $d$.
%
%
It is well known that containment for regular languages is $\PSPACE$-complete~\cite{stockmeyer1973word}, even for restricted classes of regular expressions~\cite{martens2009complexity}.
Since our expressions are extensions of regular expressions and automata, these results imply that a  $\PSPACE$ bound is the best we can aim for.
Given that the complexity of evaluation and satisfiability for $\VA$ increases compared to regular languages, one would expect the complexity of containment to do the same. 
Fortunately, this is not the case. In fact,
containment of all infromation extraction languages we consider is $\PSPACE$-complete.
\begin{theorem}\label{theo:contain-general}
	$\containment$ of extraction rules and $\containment$ of $\VA$ are both $\PSPACE$-complete.
\end{theorem}

Given that all RGX subfragments contain regular expressions, it does not make sense to consider the functional or sequential restrictions of RGX to lower the complexity.
Instead, we have to look for subclasses of regular languages where containment can be decided efficiently like, for example, deterministic finite state automata~\cite{HopcroftU79}.
It is well-known that containment between deterministic finite state automata can be checked in $\PTIME$~\cite{stockmeyer1973word}.
Then a natural question is: what is the deterministic version of $\VA$? 
One possible approach is to consider a deterministic model that, given any document produces a mapping deterministically. 
Unfortunately, this idea is far too restrictive since it will force the model to output at most one mapping for each document.
A more reasonable approach is to consider an automata model that behaves deterministically \emph{both} in the document and the mapping.
This can be formalised as follows: a $\VA$ $(Q, q_0, q_f, \delta)$ is \emph{deterministic} if for every $p \in Q$ and $v \in \Sigma \cup \{\Open{x}, \Close{x} \mid x \in \VV \}$ there exists at most one $q\in Q$ such that $(p, v, q) \in \Delta$.
That is, the transition relation of a deterministic $\VA$ is a  function with respect to both $\Sigma$ and $\VV$. 
Although the deterministic version of $\VA$ seems straightforward, as far as we know, this is the first attempt to introduce this notion for infromation extraction languages. 

The first natural question to ask is whether deterministic $\VA$ can still define the same class of mappings as the non-deterministic version.
Indeed, one can easily show that every $\VA$ can be determinised by following the standard determinisation procedure~\cite{HopcroftU79}.
\begin{proposition}\label{prop:determinize}
	For every $\VA$ $\cA$, there exists a deterministic $\VA$ $\cA^{\text{det}}$ such that $\semd{\cA} = \semd{\cA^{\text{det}}}$ for every document $d$.
\end{proposition}
As mentioned previously, the motivation of having a deterministic model is to look for subclasses of $\VA$ where $\containment$ has lower complexity. 
We can indeed show that this is the case for deterministic \VA, although the drop in complexity is not as dramatic as with regular languages.
\begin{theorem}\label{theo:contain-det}
   $\containment$ of deterministic $\VA$ is in $\pitwop$. Moreover, $\containment$ of deterministic sequential $\VA$ is $\coNP$-complete.
\end{theorem}

Although containment of deterministic models is better than in the general case, the complexity is still high. 
By taking a closer look at the lower bound (see the appendix), this happens because of the following two reasons: (i) some mappings extract spans that \emph{intersects} at extreme points; and (ii) the automaton can open a variable, but never close it. Notice that the problem (ii) is solved by sequential $\VA$. To overcome (i) we need the following definition. We say that two spans $(i_1, j_1)$ and $(i_2, j_2)$ are \emph{point-disjoint} if $\{i_1, j_1\} \cap \{i_2, j_2\} = \emptyset$, and we say that a mapping $\mu$ is point-disjoint if the images of different variables are point-disjoint. A $\VA$ automaton is point-disjoint if all mappings in $\semd{\gamma}$ are point-disjoint for every document $d$.
Using these restrictions we  can show tractability of containment.
\begin{theorem}\label{theo:contain-point-disjoint}
 $\containment$ of deterministic sequential $\VA$ that produce point-disjoint mappings is in $\PTIME$. 
\end{theorem}

\section{Conclusions}\label{sec-conc}

In this paper we propose to extend the semantics of several previously proposed IE formalisms with mappings in order to support extraction of information that is potentially incomplete. This approach allows us to simplify and make fully declarative the semantics of regex formulas of \cite{FKRV15} and extraction rules of \cite{AMRV16}, while at the same time making it possible to compare their expressive power.
From our analysis it follows that several variants of expressions proposed by \cite{FKRV15} and \cite{AMRV16} are in fact equivalent, and that obtaining an efficient algorithm for enumerating all of their outputs is generally not possible. To overcome the latter, we isolate a class of sequential regex formulas, which extend the functionality constraint of \cite{FKRV15}, and show that these can be efficiently evaluated both in isolation, and when combined into tree-like rules of~\cite{AMRV16}. Finally, the good properties of sequential formulas and tree-like rules are also preserved when considering main static tasks, thus suggesting that they have the potential to serve as a theoretical base of information extraction languages.

\medskip
\noindent{\bf Acknowledgements.} The authors were supported by the Nucleus Millennium Center for Semantic Web Research grant N120004. Vrgo\v{c} was also supported by the FONDECYT project nr. 11160383.



\bibliographystyle{plain}
\bibliography{biblio}

\newpage

\onecolumn
\appendix


\label{sec:appendix}
\section{DEFINITIONS}
\medskip

\subsection*{Extended Definition for Variable Automata}

For the proofs in this appendix we will use an equivalent, though more precise,
definition for variable automata. This definition is an adaption of the original
definition given in \cite{FKRV15} that allows for mappings instead of total
functions.

\smallskip
\noindent{\bf Variable-stack automaton.} This class of automata operates in a way analogous to \RGX; that is, it behaves as a usual finite state automaton, except that it can also open and close variables. To mimic the way this happens in \RGX, variable-stack automata use a stack in order to track which variables are opened, and when to close them.

Formally, a \textit{variable-stack automaton} (\(\VAStk\)) is a tuple \((Q, q_0, q_f, \delta)\),
where: \(Q\) is a finite set of \textit{states}; \(q_0 \in Q\) is the
\textit{initial state}; \(q_f \in Q\) is the \textit{final state}; and \(\delta\)
is a \textit{transition relation} consisting of triples of the
forms \((q, w, q')\), \((q, \epsilon, q')\), \((q, \Open{x}, q')\) or \((q,
\Close{}\,,q')\), where \(q, q' \in Q\), \(w \in \Sigma\), \(x \in \VV\),
\(\vdash\) is a special \textit{open} symbol, and \(\dashv\) is a special
\textit{close} symbol. For a $\VAStk$ automaton $A$ we define the set $\Var(A)$ as the set of all
variables $x$ such that $\Open{x}$ appears in some transition of $A$.

A \textit{configuration} of a \(\VAStk\) automaton \(A\) is a tuple \((q, V, Y, i)\), where \(q
\in Q\) is the \textit{current state}; \(V \subseteq \Var(A)\) is the stack of
\textit{active variables}; \(Y \subseteq \Var(A)\) is the set of \textit{available
  variables}; and \(i \in [1, |d| + 1]\) is the \textit{current position}.
A \textit{run} \(\rho\) of \(A\) over document \(d = a_1a_2 \cdots a_n\) is a
sequence of configurations \(c_0, c_1, \ldots, c_m\) where \(c_0 = (q_0,
\emptyset, \Var(A), 1)\) and for every \(j \in [0, m - 1]\), one of the following
holds for \(c_j = (q_j, V_j, Y_j, i_j)\) and \(c_{j + 1} = (q_{j + 1}, V_{j +
  1}, Y_{j + 1}, i_{j + 1})\):
\begin{enumerate}\itemsep=0pt
\item \(V_{j + 1} = V_j\), \(Y_{j + 1} = Y_j\), and either
  \begin{enumerate}\itemsep=0pt
  \item \(i_{j + 1} = i_j + 1\) and \((q_j, a_{i_j}, q_{j + 1}) \in \delta\)
    (ordinary transition), or
  \item \(i_{j + 1} = i_j\) and \((q_j, \epsilon, q_{j + 1}) \in \delta\)
    ($\varepsilon$-transition).
  \end{enumerate}
\item \(i_{j + 1} = i_j\) and for some \(x \in \Var(A)\), either
  \begin{enumerate}\itemsep=0pt
  \item \(x \in Y_j\), \(V_{j + 1} = V_j \cdot x\), \(Y_{j + 1} = Y_j
    \setminus \{x\}\), and \((q_j, \Open{x}, q_{j + 1}) \in \delta\) (variable
    insert), or
  \item \(V_j=V_{j+1}\cdot x\), \(Y_{j + 1} = Y_j\)
    and \((q_j, \Close{}\,, q_{j + 1}) \in \delta\) (variable pop).
  \end{enumerate}
\end{enumerate}
The set of runs of \(A\) over a document \(d\) is denoted \(\runs(A, d)\).
A run \(\rho = c_0, \ldots, c_m\) is \textit{accepting} if \(c_m = (q_f, V_m,
Y_m, |d| + 1)\). The set of accepting runs of \(A\) over \(d\) is denoted
\(\aruns(A, d)\). Let \(\rho \in \aruns(A, d)\), then for each variable \(x \in
\Var(A) \setminus (Y_m \cup V_m)\) there are configurations \(c_b = (q_b, V_b, Y_b,
i_b)\) and \(c_e = (q_e, V_e, Y_e, i_e)\) such that \(V_b\) is the first one in
the run where \(x\) occurs and \(V_e\) (with \(e \neq m\)) is the last one in
the run where \(x\) occurs; the span \((i_b, i_e)\) is denoted by \(\rho(x)\).
The mapping \(\mu^{\rho}\) is such that \(\mu^{\rho}(x)\) is \(\rho(x)\)
if \(x \in \Var(A) \setminus (Y_m \cup V_m)\), and undefined otherwise. Finally, the semantics of $A$ over $D$, denoted by \(\semd{A}\)
is defined as the set \(\{\mu^{\rho} \mid \rho \in \aruns(A, d)\}\).

Note here that the only difference between our definition and \cite{FKRV15} is how we define accepting runs and the mappings $\mu^\rho$. In particular, we do not impose that all the variables in $\Var(A)$ should be used in the run, and we also allow some of them to remain on the stack. Furthermore, we leave our mappings undefined for any unused variable.

\smallskip
\noindent{\bf Variable-set automaton.} Following \cite{FKRV15} we introduce a more general class of automata which allow defining mappings that are not necessarily hierarchical as in the case of \VAStk\ automata and \RGX. We call these automata {\em variable-set automata} (\(\VA\)). The definition of variable-set automata is almost identical to the one of $\VAStk$ automata, but we now have transitions of the form \((q,\Close{x}, q')\) instead of \((q,
\Close{}\,, q')\), that allow us to explicitly state which variable is closed. Likewise, instead of a stack, they operate using a {\em set}, thus allowing us to add and remove variables in any order.
The only difference between \VA\ and \VAStk\ automata is in the condition 2.(b) of a run, where we directly stipulate which variable should be removed from the set $V_j$ (this used to be a stack in \VAStk). Acceptance is defined analogously as before. To avoid repeating the same definition we refer the reader to \cite{FKRV15} for details, taking note of the new semantics.

\section{PROOFS FROM SECTION~\secref{sec:comp}}
\medskip

\subsection*{Proof of Theorem~\ref{th-frgx}}


The definition presented in this paper and the one in \cite{FKRV15} are
syntactically identical. Therefore, it only remains to show that their semantics
are equivalent. The semantics of Regex formulas in \cite{FKRV15} are defined
using the notion of \emph{parse trees}. Given a formula $\gamma$ and a document
$d$, a \emph{$\gamma$-parse} is a tree where the internal nodes correspond to
operators and variables according to the structure of $\gamma$ while the leaves
correspond to alphabet symbols that compose the document $d$. It is
straightforward to prove that, in the functional case, both definitions are
equivalent since there is a direct correspondence between the subtree rooted at
an specific node, and the first component of the tuples in $\sempd{\gamma}$.
\hfill\qed\medskip

\subsection*{Proof of Theorem~\ref{theo:equiv-SRE}}

The definition of the semantics of span regular expressions in \cite{AMRV16} is
similar to the one presented in this paper, except for three aspects: (1) the
definition is based on \emph{total} functions (instead of mappings), (2)
variables which are not given an specific value can take any value, and (3)
expressions of the form $x\{\Sigma^*\} \cdot x\{\Sigma^*\}$ are satisfiable. By
considering only \textit{proper} expressions we address (3). By letting
$\semd{\gamma}' = M \Join \semd{\gamma}$ we address (1) because $M$ \emph{only}
contains total functions (and so does $\semd{\gamma}'$ as a consequence), and we
address (2) because $M$ contains \emph{all} total functions and, therefore,
unassigned variables from mappings in $\semd{\gamma}$ will be given all possible
values.
\hfill\qed\medskip

\subsection*{Proof of Theorem~\ref{th-stack_rgx}}


This proof is a generalisation of Theorem~4.4 presented in \cite{FKRV15} to the setting supporting mappings. Here we present a sketch of the original proof, along with the necessary modifications to adjust to our more general semantics.


First, we show that every \RGX\ has an equivalent \VAStk. This can be proved by
adapting the well-known \emph{Thompson's Construction Algorithm}
\cite{HopcroftU79}, that takes a regular expression as input and constructs an
equivalent automaton. The only difference is that we extend the algorithm to
handle expressions of the form $x\{\gamma\}$ by respectively adding an open and
close transition for variable $x$ connected to the initial and final states of
the automaton constructed for $\gamma$. It is straightforward to prove by
induction over the structure of regex formulas that the constructed automaton
will be equivalent to the input expression.

For the opposite direction, we can use the \emph{state elimination} technique
\cite{HopcroftU79}. This technique consists in allowing transitions to be
labeled with regular expressions and eliminating states by replacing them with
equivalent transitions (see Figure~\ref{fig:state-remove}). Let $A = (Q, q_0,
q_f, \delta)$ be the input vstk automaton.

First, we add to $A$ the necessary $\varepsilon$-transitions so that the
incoming transitions of each state either: are all variable transitions, or
contain no variable transitions. Then, using the aforementioned technique, we
remove all states except for the initial state, the final state, and all those
that have incoming variable transitions (we assume that the final state has no
incoming variable operations). Notice that after this, every transition will be
associated with exactly one variable operation (except for the transitions that
end in the final state). This is what \cite{FKRV15} denominates a
\emph{vstk-graph automaton}. Let this resulting vstk-graph automaton be $A'$.

Second, we will construct a new automaton by considering paths in $A'$ that go
from the initial state to the final state. This step has the main difference
with the original proof because we will consider all paths of length \emph{at
  most} $2 \cdot |\VV| + 1$, whereas the original considers only those which are
exactly of that length. This is because the original proof considered only
functional paths, i.e.\ those that open and close all variables, and that
therefore use exactly $2 \cdot |\VV|$ variable operations. For each path with
the aforementioned characteristics, we build a new automaton that consists
solely of that path, called \emph{vstk-path automaton} in \cite{FKRV15},
resulting in a set of such automata. At this step, we can easily remove in each
path the variable operations that open a variable but never close it again.
(Remember that valid runs may open a variable and never close it. The result,
however, is the same as if the variable was not opened.) The new automaton $A''$
is constructed by merging the initial states and the final states of all the
vstk-path automata, resulting in what \cite{FKRV15} calls a \emph{vstk-path
  union automaton}.

Finally, it is very easy to see how to obtain a \RGX\ that is equivalent to a
certain vstk-path automaton: if the vstk-path has a valid run, then we simply
concatenate the labels of the transitions, replacing $\Open{x}$ for $x\{$ and
$\Close{}$ for $\}$. Therefore, the final \RGX\ is that which corresponds to the
disjunction of the \RGX\ equivalent to each of the vstk-path automata in~$A''$.

It is not difficult to prove that the final \RGX\ will equivalent to $A$, since
it is clear from the semantics of \RGX\ and \VAStk\ that each of the steps will
preserve the equivalency of the expressions.
\hfill\qed\medskip

\begin{figure}
  \centering
  \begin{tikzpicture}[->, auto, font=\scriptsize, node distance=1cm]
    \node[state] (A)                  {$q$};
    \node[state] (B) [right=of A]     {$s$};
    \node[state] (C) [right=of B]     {$p$};
    \node        (D) [right=5mm of C] {\large$\Rightarrow$};
    \node[state] (E) [right=5mm of D] {$q$};
    \node[state] (F) [right=3cm of E] {$p$};

    \path
    (A) edge              node [below] {$Q$} (B)
        edge [bend left]  node         {$R$} (C)
    (B) edge [loop below] node         {$S$} (B)
        edge              node [below] {$P$} (C)
    (E) edge              node         {$R \vee (Q \cdot S^* \cdot P)$} (F)
    ;
  \end{tikzpicture}
  \caption{Example of the elimination of state $s$ in the \emph{state
      elimination} technique (as shown in \cite{FKRV15}).}
  \label{fig:state-remove}
\end{figure}
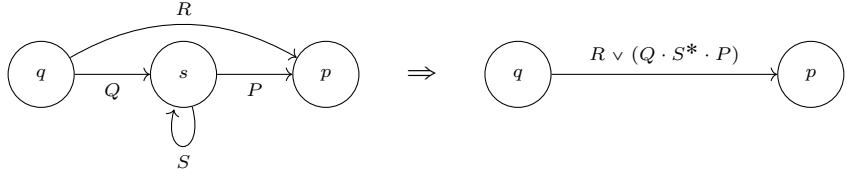

\subsection*{Proof of Theorem~\ref{th-rgx_equiv_hierarchical_va}}

This proof is a generalisation of Theorem~4.6 presented in \cite{FKRV15} to the setting supporting mappings. Here we present a sketch of the original proof, along with the necessary modifications to adjust to our more general semantics.


From the proof of Theorem~\ref{th-stack_rgx}, it is very clear that each \RGX\
has an equivalent \VASet\ automaton, since the construction procedure described
can be trivially adapted to \VASet\ automata.

To show that every hierarchical \VASet\ automaton has an equivalent \RGX, we
will use the same vstk-path union construction from the previous proof (which
will result in a \emph{vset-path union automaton} in this case). Let $A$ be the
input vset automaton and let $A'$ be the resulting vset-path union automaton. It
can be proven, without much difficulty, that if $A$ is hierarchical, then $A'$
will be hierarchical. It is proved in \cite{FKRV15} that if $A'$ is
hierarchical, then its variable operations can be reordered so that they are
``correctly nested''. After this reordering, a \RGX\ can be obtained from $A'$
in the same manner than the previous proof.
\hfill\qed\medskip

\subsection*{Proof of Theorem~\ref{th-va_algebra_equiv}}


This proof is a generalisation of Theorem~4.14 presented in \cite{FKRV15} to the setting supporting mappings. Here we present a sketch of the original proof, along with the necessary modifications to adjust to our more general semantics.


We expand the theorem to separate the proof into two containments and an equivalence:
\begin{equation*}
  \VAStk^{\{\cup,\pi,\Join\}} \subseteq \VA^{\{\cup,\pi,\Join\}} \equiv \VA \subseteq \VAStk^{\{\cup,\pi,\Join\}}
\end{equation*}
The first containment follows from Theorem~\ref{th-rgx_equiv_hierarchical_va}.

The equivalence can be proved as follows. Unions can be simulated in \(\VA\) by
simply using \(\epsilon\)-transitions at the start of different automata.
Projections can be simulated by using the path-union automata construction and
replacing the variable transitions of the projected variables with
\(\epsilon\)-transitions. Joins are simulated in a similar way to NFA
intersections: by constructing an automaton that runs both automata ``in
parallel'', taking care of opening and closing shared variable at the same time.
For this to work, however, the automata need to first be transformed into
\emph{lexicographic} \(\VA\). These work the same way as \(\VA\), but guarantee
that for every document \(d\) and mapping \(\mu\) accepted by an automaton
\(A\), there is an accepting run in \(A\) that performs the variable operations
needed to produce \(\mu\) in a specific order.

We finish by explaining why the second containment holds. As usual, the
path-union automata construction can be used to simplify the proof, meaning that
we only need to consider path automata. A path \(\VA\) can be simulated by a
\(\VAStk\) by introducing auxiliary variables that split spans which may have
been non-hierarchical, using joins to ensure that these auxiliary variables
correspond to the original variables, and projecting away the auxiliary
variables. For the full construction, refer to the proof of Lemma~4.13 in
\cite{FKRV15}.
\hfill\qed\medskip

\medskip

\subsection*{Proof of Theorem~\ref{prop-rgx-rules}}

First we will show that there is an extraction rule that has no equivalent
\(\RGX\). As shown in \cite{FKRV15}, functional \(\RGX\) are hierarchical. It is
clear that this also extends to non-functional \(\RGX\). With this mind, it is
easy to realize that the extraction rule \(x \wedge x.\Sigma^{*} \cdot y \cdot
\Sigma^{*} \wedge x.\Sigma^{*} \cdot z \cdot \Sigma^{*}\) is not hierarchical,
since \(y\) and \(z\) might be assigned spans that overlap in a non-hierarchical
way. This rule, therefore, cannot be expressed by a \(\RGX\).

Now we prove that there is a variable regex that has no equivalent extraction
rule. Consider the following variable regex: \(\gamma = (a \cdot x\{b\}) \vee (b
\cdot x\{a\})\). There are only two ways in which a document and mapping can
satisfy it: (1) \(d_1 = ab\) and \(\mu_1(x) = (2, 3)\); or (2) \(d_2 = ba\) and
\(\mu_2(x) = (2, 3)\). Suppose that there is an extraction rule \(\varphi\) that
is satisfied only by these two document-mapping pairs. By the structure of
extraction rules, we know that there is an extraction expression \(x.\varphi_x\)
such that \(\varphi_x\) is equivalent to the expression \(a \vee b\); if not, we
can construct a document \(d_3\) that satisfies \(\varphi\) and is different
from \(d_1\) and \(d_2\). By the same reason, we know that \(\varphi_0\), the
root extraction expression of \(\varphi\), must be equivalent to \(ax \vee bx\).
Notice, however, that the document \(d_3 = aa\) and the mapping \(\mu_3\) such
that \(\mu_3(x) = (2, 3)\) satisfy \(\varphi\). We have reached a contradiction,
and therefore conclude that such \(\varphi\) does not exist. \hfill\qed\medskip

\subsection*{Proof of Theorem~\ref{thm:simple_to_dag}}


Consider an arbitrary simple rule that is {\em functional}. We start by
analyzing the sort of values that a mapping can assign to the variables which
form a cycle. For this, take any rule $\varphi$ and assume that there is a
simple cycle $x_1,\ldots ,x_n$ appearing in $G_\varphi$ and a mapping $\mu$
satisfying $\varphi$. Then the following must hold:
\begin{enumerate}
\item {\em All variables in the cycle must be assigned the same value.} This
  follows from the fact that in a simple rule each edge $(x,y)$ in $G_\varphi$
  implies that $\mu(x)$ contains $\mu(y)$ (see Figure~\ref{fig:cycle-simple}).
\item {\em Every variable reachable from a cycle, but not inside it, must be
    assigned the empty content.} This follows from the observation above, plus
  the fact that edges $(x,y)$ and $(x,z)$ in $G_\varphi$ imply that $x$ and $y$
  appear in the same \SRGX. By the structure of \SRGX, if $x \neq y$ then
  $\mu(y)$ and $\mu(z)$ must be disjoint (see Figure~\ref{fig:cycle-reachable}).
\item {\em If the cycle has a chord, then all the variables inside it must be
    assigned the empty content.} Here a chord means that we have a path from
  some $x_i$ to some $x_j$ inside $G_\varphi$ which consists of nodes not
  belonging to our cycle, or there is a direct edge between them which is not
  part of the cycle. In the case there is an intermediate node, we know that it
  must be assigned $\varepsilon$, therefore $x_j$ and all other nodes in the
  cycle must be $\varepsilon$ as well. If the edge is direct, then by the
  definition of a chord, $x_j$ is not a successor of $x_i$ in the cycle, so just
  as in the previous case, the content of the successor of $x_i$ and the content
  of $x_j$ must be disjoint and equal, which is only possible if they are
  $\varepsilon$ (see Figure~\ref{fig:cycle-chord}).
\end{enumerate}

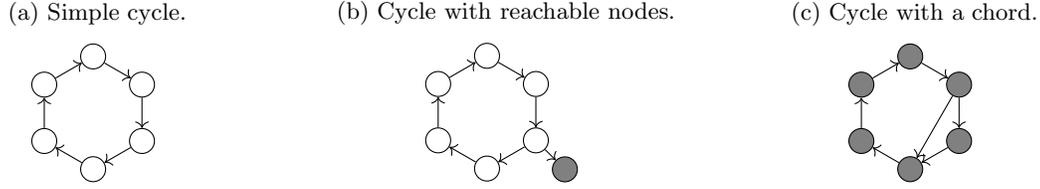
\begin{figure}
\caption{Different cycle arrangements.}
\centering
\begin{subfigure}[b]{0.3\textwidth}
  \caption{Simple cycle.}
  \label{fig:cycle-simple}
  \centering
  \begin{tikzpicture}[
    node distance=.2cm,
    nodes={draw, circle},
  ]
  \graph [clockwise, radius=.75cm, ->, empty nodes] {
    subgraph C_n [n=6];
  };
\end{tikzpicture}
\end{subfigure}
\begin{subfigure}[b]{0.3\textwidth}
  \caption{Cycle with reachable nodes.}
  \label{fig:cycle-reachable}
  \centering
  \begin{tikzpicture}[
    node distance=.2cm,
    nodes={draw, circle},
    mark/.style={fill=black!50},
    mark/.default=,
  ]
  \graph [clockwise, radius=.75cm, ->, empty nodes] {
    subgraph C_n [n=6];
  };
  \node[mark, below right=of 3] (a) {};
  \path[->] (3) edge (a);
\end{tikzpicture}
\end{subfigure}
\begin{subfigure}[b]{0.3\textwidth}
  \caption{Cycle with a chord.}
  \label{fig:cycle-chord}
  \centering
  \begin{tikzpicture}[
    node distance=.2cm,
    nodes={draw, circle},
  ]
  \graph [clockwise, radius=.75cm, ->, empty nodes, nodes={fill=black!50}] {
    subgraph C_n [n=6];
    2 -> 4;
  };
\end{tikzpicture}
\end{subfigure}
\end{figure}

The procedure for eliminating cycles from simple rules is based on the following
colouring scheme for a graph $G_{\varphi}$ associated with the rule $\varphi$.
Let $\varphi$ be an extraction rule with variables $x_1,\ldots x_n$. We will
colour a node $x_i$ {\em black} if:
\begin{itemize}
\item $x_i.\varphi_i$ appears in $\varphi$ and $\varphi_i$ is such that, when
  treating it as a regular expression, every word that can be derived from it
  must contain a symbol from \(\Sigma\).
\end{itemize}
We then paint the graph by assigning the colour {\em red} to all black nodes,
and all nodes which can reach a black node. All other nodes are coloured {\em
  green}. It is clear that this procedure can be carried out in polynomial time,
since reachability takes only polynomial time. Note that in a black node coming
from a conjunct of the form $x_i.\varphi_i$, the content of each variable
appearing in $\varphi_i$ must be {\em strictly} contained in the content of the
variable $x_i$. This is because $\varphi_i$ is functional and, since we painted
its node black, it must have symbols from $\Sigma$ which are not part of the
content of the variables used in $\varphi$. Also note that each cycle has to be
coloured using the same colour.

If we now have a simple cycle $x_1,\ldots ,x_n$ we can eliminate it by
considering its colour:
\begin{itemize}
\item {\em If the cycle is coloured red, then the rule is not satisfiable, so we
    can replace it by an arbitrary unsatisfiable dag-like rule.} We have two
  cases here. First, if a cycle contains a black node, then the content of its
  successor must be strictly contained inside its own content, which cannot
  happen by the analysis above. Second, if a node $x$ in the cycle can reach
  some black node not inside the cycle, then its content must be different from
  $\varepsilon$, which contradicts point (2) of the above analysis.
\item {\em If the cycle is green we can simplify it using an auxiliary
    variable.} Let $u_1, \ldots, u_m$ be the variables that are not part of the
  cycle and for which there is an edge $(u_i, x_i)$. Let $y_1, \ldots, y_l$ be
  the variables that are not part of the cycle and are reachable from some $x_i$
  (they must have empty content, as proved before). We then add an auxiliary
  variable $w$ and an edge from it to $x_1$. Each expression associated with
  some $u_i$ is changed so that it uses $w$ instead of $x_i$, and all expression
  associated with some $y_i$ are changed to $y_i.\varepsilon$. Next, for $i<n$,
  an expression $x_i.\varphi_i$ is changed to $x_i.\varphi_i'$, where
  $\varphi_i'$ maintains the possible orderings of variables according to
  $\varphi_i$. This is done by removing all other letters or starred
  subexpressions, and is explained in detail later in this proof. For $x_n$, we
  replace the occurrences of $x_1$ by $\Sigma^*$. This yields an equivalent
  \emph{simple} rule without the mentioned cycle.
\end{itemize}

As an example of how the rewriting above works, consider the rule $x.y \wedge
y.z \wedge z.ux.$ This rule can be rewritten to $w.x \wedge x.y \wedge y.z
\wedge z.u \cdot \Sigma^* \wedge u.\varepsilon$ by introducing the auxiliary
variable $w$, forcing the variable $u$ to have empty content, and breaking the
cycle at $z$.

Of course, here we explained only how a single cycle can be removed, but how do
we transform a rule with multiple cycles in its graph? For this we start by
identifying the strongly connected components of our graph $G_\varphi$. Each
component can then be either: (a) a single node, (b) a simple cycle, or (c) a
simple cycle with additional edges. In the latter two cases, if any component is
coloured red, we know that the rule is unsatisfiable, so we can replace it by an
arbitrary unsatisfiable dag-like rule. In the case they are coloured green, we
can deploy the procedure above to remove the cycles, taking care that in the
case (b) our variables can take an arbitrary, but always equal value, while in
the case (c) they must be equal to the empty content. In both cases, all the
variables reachable from the component are made $\varepsilon$.

Now we will precisely describe the procedure for eliminating cycles in rules.
Let \(\varphi = \varphi_{x_0} \wedge x_1.\varphi_{x_1} \wedge \cdots \wedge
x_m.\varphi_{x_m}\) be a simple rule such that each \(\varphi_{x_i}\) (\(0 \leq
i \leq m\)) is functional, and let \(G_{\varphi}\) be its graph. We will assume
that for every variable \(x \in \Var(\varphi)\) there is an extraction
expression \(x.\varphi_x\) in \(\varphi\); if not, we can simply add the
extraction expression \(x.\Sigma^{*}\). We will now describe in detail the
procedure that produces an equivalent dag-like rule \(\alpha\).

First, we will colour the nodes in \(G_{\varphi}\). For this, we define a
function \(\nu: \SRGX \to \SRGX\) that will indicate when a variable cannot have
empty content. Here, \(\emptyset\) has the usual definition in the regular
expression context, with the following properties: \(\emptyset \cdot \alpha =
\emptyset\), \(\emptyset \vee \alpha = \alpha\), \(\emptyset^* = \varepsilon\),
where \(\alpha\) is any expression.
\begin{itemize}
\item \(\nu(a) = \emptyset\), where \(a \in \Sigma\).
\item \(\nu(x) = x\), where \(x \in \VV\).
\item \(\nu(\varphi_1 \cdot \varphi_2) = \nu(\varphi_1) \cdot \nu(\varphi_2)\).
\item \(\nu(\varphi_1 \vee \varphi_2) = \nu(\varphi_1) \vee \nu (\varphi_2)\).
\item \(\nu(\varphi^{*}) = \epsilon\).
\end{itemize}
Thus, we paint a node \(x_i\) \textit{black} if \(\nu(\varphi_i) = \emptyset\).
After this, we paint a node \textit{red} if it is black or if it can reach a
black node. We do this by painting black nodes red and then ``flooding'' the
graph by doing depth-first search from black nodes using the edges in reverse.

Second, we run \textit{Tarjan's Strongly Connected Components Algorithm}
(\cite{Tarjan72}). This algorithm will compute the strongly connected components
(SCCs) in the graph and output them in topological order with respect to the dag
formed by the SCCs. We denote the ordered SCCs as \(S_1, \ldots, S_l\), where
each \(S_i\) (\(1 \leq i \leq l\)) is a set of nodes.

Finally, we process the SCCs in order. Each SCC \(S_i\) will be of one of the
following types: (1) \(S_i\) is a single node; (2) \(S_i\) is a simple cycle; or
(3) \(S_i\) contains a cycle and has additional edges (this includes everything
that does not fall under types (1) and (2)). Notice that the type of \(S_i\) can
easily be computed in polynomial time. Now, according to the type, do the
following:
\begin{itemize}
\item Type (1): let \(S_i = \{y\}\). We copy the extraction expression
  \(y.\varphi_y\) to \(\alpha\).
\item Type (2): let \(S_i = \{y_1, \ldots, y_k\}\), such that \((y_k, y_1)\) and
  \((y_j, y_{j + 1})\) are edges in \(G_{\varphi}\), for \(j \in [1, k - 1]\).
  If \(S_i\) has a red node, then the rule is unsatisfiable and we may stop and
  replace \(\alpha\) with any unsatisfiable dag-like rule. Otherwise, we add a
  new auxiliary variable \(u_i\) and replace every appearance of variables of
  \(S_i\) in \(\alpha\) with \(u_i\). Add the following extraction expressions
  to \(\alpha\):
  \begin{itemize}
  \item \(u_i.y_1\);
  \item \(y_j.\nu(\varphi_{y_j})\), for \(j \in [1, k - 1]\);
  \item and \(y_k.\psi\), where \(\psi\) is \(\nu(\varphi_{y_k})\) with all
    appearances of \(y_1\) replaced with \((\Sigma^{*})\).
  \end{itemize}
  After this, mark every SCC reachable from \(S_i\) as a type (3) SCC.
\item Type (3): let \(S_i = \{y_1, \ldots, y_k\}\). If \(S_i\) has a red node,
  then the rule is unsatisfiable and we may stop and replace \(\alpha\) with any
  unsatisfiable dag-like rule. Add an auxiliary variable \(u_i\) and add the
  following rules to \(\alpha\):
  \begin{itemize}
  \item \(u_i.y_1 \cdots y_k\);
  \item \(y_j.\psi\), for \(j \in [1, k]\) where \(\psi\) is
    \(\nu(\varphi_{y_j})\) with all appearances of variables \(y_1, \ldots,
    y_k\) replaced with \(\epsilon\).
  \end{itemize}
  After this, mark every SCC reachable from \(S_i\) as a type (3) SCC.
\end{itemize}

The resulting rule \(\alpha\) will be dag-like and equivalent to \(\varphi\). If
we take into account the observations presented at the beginning of this proof,
then it is straightforward to verify that the transformations outlined above
will remove the cycles in \(G_{\varphi}\) while preserving equivalence.
\hfill\qed\medskip

\subsection*{Proof of Proposition~\ref{prop:nonfunc}}


Let \(\varphi = \varphi_0 \wedge x_1.\varphi_1 \wedge \cdots \wedge
x_m.\varphi_m\) be a rule such that each \(\varphi_i\) is a \(\SRGX\), where \(i
\in [0, m]\). We can transform each \(\varphi_i\) into an equivalent disjunction
\(\psi_{i, 1} \vee \cdots \vee \psi_{i, l_i}\) where each \(\psi_{i, j}\) is a
functional \(\SRGX\). This is done by using the \(\PUStk\) construction from
Theorem~\ref{th-stack_rgx}, originally presented in \cite{FKRV15}. Specifically,
we transform \(\varphi_i\) into a \(\VAStk\) \(A\) and then into a \(\PUStk\)
\(A'\). It is clear that each path in \(A'\) can be directly transformed into a
functional \(\SRGX\) (since paths do not have disjunctions of variables).
Therefore, each \(\psi_{i, j}\) will be a functional \(\SRGX\). Notice, however,
that this transformation might produce exponentially many \(\psi_{i, j}\) with
respect to the size of \(\varphi_{i}\).

As an example of this step, the \(\SRGX\) \((x \vee y)\cdot(z \vee w)\) is
equivalent to the disjunction \((\epsilon \vee x \cdot z \vee x \cdot w \vee y
\cdot z \vee y \cdot w)\). Note that each of the disjuncts is independently
functional.

Rule \(\varphi\) will be equivalent to the set of rules that consist of all
possible conjunctions that can be made by taking one disjunct \(\psi_{i, j}\)
from every extraction expression (\(i \in [0, m]\)). Formally, \(\varphi\) will
be equivalent to \(\{\psi_{0, k_0} \wedge x_1.\psi_{1, k_1} \cdots \wedge
x_m.\psi_{m, k_m} \mid (k_0, \ldots, k_m) \in [1, l_0] \times \cdots \times [1,
l_m]\}\). Note that this will produce another exponential blow-up in size. The
resulting set will therefore be double-exponential in size with respect to
\(\varphi\).

For example, consider the rule \(\varphi = (x \vee y) \wedge x.(a \vee b) \wedge
y.(c)\). Then, \(\varphi\) is equivalent to the following set of rules:
\begin{equation*}
\{x \wedge x.a \wedge y.c,\;
x \wedge x.b \wedge y.c,\;
y \wedge x.a \wedge y.c,\;
y \wedge x.b \wedge y.c \}
\end{equation*}

Now we prove that the transformation is correct. The correctness of the
transformation from \(\SRGX\) to \(\PUStk\) carries from the original proof
without much modification. Given the definition of the semantics for rules, it
is fairly easy to observe that taking every possible combination of the
disjuncts in each extraction expression will produce an equivalent set of rules.

Finally, by applying Theorem~\ref{thm:simple_to_dag}, we can transform this
union of functional rules into a union of functional dag-like rules.
\hfill\qed\medskip

\subsection*{Proof of Proposition~\ref{thm:dag_to_tree}}


In order to prove this proposition, we will first state and prove two auxiliary
lemmas which will be necessary for this proof.

\begin{lemma}\label{lemma:tl-rgx}
  Every tree-like expression can be transformed into an equivalent \RGX.
\end{lemma}

\begin{proof}
We will transform tree-like extraction rules into \(\RGX\) by recursively
nesting extraction expressions into their associated variables. The procedure is
as follows. Let \(\theta = \varphi_{x_0} \wedge x_1.\varphi_{x_1} \wedge \cdots
\wedge x_m.\varphi_{x_m}\) be a tree-like extraction rule, and let
\(G_{\theta}\) be its graph. Without loss of generality, we assume that every
variable \(x \in \Var(\theta)\) appears on the left side of an extraction
expression (if not, we can add \(x.\Sigma^*\)). For all \(i \in [0, m]\) we
define a \(\RGX\) \(\gamma_{x_i}\) as \(\varphi_{x_i}\) where each mention of
variable \(y\) is replaced with \(y\{\gamma_y\}\). The expression
\(\gamma_{x_0}\) will be a well-formed \(\RGX\) and equivalent to \(\theta\). It
is straightforward to prove this last statement by induction.

As an example, consider the tree-like rule \(\varphi = (a \cdot x \cdot b \cdot
y) \wedge x.(abc \cdot z) \wedge y.(\Sigma^*) \wedge z.(d)\). The resulting
\(\RGX\) in this case would be \(\gamma = a \cdot x\{abc \cdot z\{d\}\} \cdot b
\cdot y\{\Sigma^*\}\).

It is clear that this procedure terminates since \(G_{\theta}\) is a forest.
Note, however, that the resulting \(\RGX\) might be of exponential size with
respect to the input extraction rule, since multiple appearances of the same
variable can cause the expression to grow rapidly when the replacements are
made.
\end{proof}

\begin{lemma}\label{lemma:utlstar-rgx}
  Unions of tree-like rules and \RGX\ formulas are equivalent.
\end{lemma}

\begin{proof}
We begin by presenting vstk-graph, vstk-path, and vstk-path union, originally
defined in \cite{FKRV15} (the vset variants are defined analogously). A
\emph{vstk-graph} is a tuple \(G = (Q, q_0, q_f, \delta)\) defined as a
vstk-automaton, except that each transition in \(\delta\) is of one of the
following forms: \((q, \gamma, \Open{x}, q')\), \((q, \gamma, \Close{}, q')\),
and \((q, \gamma, q_f)\), where \(q, q' \in Q \setminus \{q_f\}\), \(x \in
\VV\), and \(\gamma\) is a regular expression over \(\Sigma\). Configurations
are defined in the same way as in the case of vstk-automata. A run \(\rho\) of
\(G\) on a document \(d\) is a sequence of configurations \(c_0, \ldots, c_m\)
where for all \(j \in [1, m - 1]\) the configurations \(c_j = (q_j, V_j,
Y_j, i_j)\) and \(c_{j +  1} = (q_{j +  1}, V_{j + 1}, Y_{j + 1}, i_{j + 1})\)
are such that \(i_j \leq i_{j + 1}\) and, depending on the transition used, one
of the following holds:
\begin{enumerate}
  \item \((q_j, \gamma, \Open{x}, q_{j + 1}) \in \delta\), the substring
    \(d(i_j, i_{j + 1})\) is in \(\LL(\gamma)\), \(x \in Y_j\), \(V_{j + 1} =
    V_j \cdot x\), and \(Y_{j + 1} = Y_j \setminus \{x\}\);
  \item \((q_j, \gamma, \Close{}, q_{j + 1}) \in \delta\), the substring
    \(d(i_j, i_{j + 1})\) is in \(\LL(\gamma)\), \(Y_j = Y_{j + 1}\), and \(V_j
    = V_{j + 1} \cdot x\); or
  \item \((q_j, \gamma, q_{j + 1}) \in \delta\) (this means \(q_{j + 1} =
    q_f\)), \(d(i_j, i_{j + 1})\) is in \(\LL(\gamma)\), \(Y_j = Y_{j + 1}\),
    and \(V_j = V_{j + 1}\).
\end{enumerate}
Accepting runs, \(\Var(G)\), and the semantics of vstk-graph, are defined the same
way as in the case of vstk-automata.

A \emph{vstk-path} \(P\) is a vstk-graph that consists of a single path. That
is, \(P\) has exactly \(m\) states \(q_1, \ldots, q_m = q_f\) and exactly \(m\)
transitions such that there is a transition from \(q_1\) to \(q_2\), from
\(q_2\) to \(q_3\), and so on. A \emph{vstk-path union} is a vstk-graph that
consists of a set of vstk-path such that: (1) each vstk-path is sequential, and
(2) every pair of vstk-paths have the same initial state, the same final state,
and share no other states.

We define \emph{path} \(\RGX\), a subset of \(\RGX\) that is simpler to analyze.
Formally, a path \(\RGX\) is an expression that can be derived using the
following grammar with \(E\) as the \emph{start symbol}.
\begin{align*}
 E &\Coloneqq x\{E\},\ x \in \VV \mid (E \cdot E) \mid R\\
 R &\Coloneqq w,\ w \in (\Sigma \cup \{\epsilon\}) \mid (R \cdot R) \mid
     (R \vee R) \mid {(R)}^{*}
\end{align*}
It is easy to see that path \(\RGX\) are equivalent to \emph{vstk-path}
automata. This is because path \(\RGX\), as vstk-path automata, do not have
disjunctions at a variable level.

With this in mind, we will show that every \(\RGX\) can be transformed into an
equivalent set of tree-like rules. It was proven in \cite[Lemma~4.3 and
Theorem~4.4]{FKRV15} that functional variable regexes are equivalent to
\textit{path union stack variable automata}, that is, stack variable automata
that consist solely of a union of disjoint paths. This result will also hold for
general \(\RGX\), with little modification to the proof. It is apparent that
each path in one of these automata will be equivalent to a path \(\RGX\), which
implies that every \(\RGX\) can be transformed into an equivalent union of path
\(\RGX\) (notice, however, that this union might be exponential in size with
respect to the starting expression).

Given this, it only suffices to show that each path \(\RGX\) is equivalent to a
tree-like rule. Let \(\gamma\) be a path \(\RGX\). Given a variable regex
\(\alpha\), we denote as \(\alpha'\) the \(\SRGX\) that results when replacing
every top-level subexpression of the form \(x\{\beta\}\) with \(x\). It is easy
to notice from the structure of \(\gamma\) that each variable can appear at most
once in the expression. Therefore, we can easily ``decompose'' \(\gamma\) into
an extraction rule by using the following procedure: add the extraction
expression \(\gamma'\) to the result and, for every subexpression of the form
\(x\{\gamma_x\}\) in \(\gamma\), add the extraction expression \(x.\gamma'_x\)
to the result. It is apparent that the resulting rule is tree-like, and it is
straightforward to prove that it is equivalent to \(\gamma\).

The proof that every set of tree-like rules can be transformed into an
equivalent \(\RGX\) follows from Lemma~\ref{lemma:tl-rgx} and the fact that
\(\RGX\)s are closed under union (by usage of the disjunction operator).
\end{proof}

With these results in mind, we now proceed to prove the proposition.

Let \(\varphi = \varphi_{x_0} \wedge x_1.\varphi_{x_1} \wedge \cdots \wedge
x_m.\varphi_{x_n}\) be a satisfiable dag-like rule such that each \(\varphi_i\)
is a functional \(\SRGX\) (\(i \in [0, n]\)), and let \(G_{\varphi}\) be its
graph. Without loss of generality, we assume that for every variable \(x \in
\Var(\varphi)\) there is an expression \(x.\varphi_x\).

Consider any pair of nodes \(x\) and \(y\) such that there are at least two
distinct paths \(u_1, \ldots, u_{l_1}\) and \(v_1, \ldots, v_{l_2}\) where \(u_1
= v_1 = x\) and \(u_{l_1} = v_{l_2} = y\) (see Figure~\ref{fig:dag-cycle}). Let
\(\mu\) be a satisfying mapping. Since all expressions are functional, we know
the following: \(\mu(x)\) contains \(\mu(u_2)\) and \(\mu(v_2)\); \(\mu(u_2)\)
and \(\mu(v_2)\) contain \(\mu(y)\); \(\mu(u_2)\) and \(\mu(v_2)\) are disjoint.
From these facts we can deduce that \(\mu(y)\) must be \(\epsilon\). Therefore,
if the rule is satisfiable, \(y\) must be painted green. Furthermore, every
variable reachable from \(y\) must be assigned \(\epsilon\) as content, which
means that \(\varphi\) may be rewritten as in the proof of
Theorem~\ref{thm:simple_to_dag} to simplify \(G_{\theta}\) for all the nodes
reachable from \(y\). This means we need only concentrate on those undirected
cycles that are ``near to the root'', since the rest can be removed in this way.

Given \(\varphi\), we first paint all nodes following the procedure from
Theorem~\ref{thm:simple_to_dag}. After this, we transform every \(\SRGX\)
\(\varphi_{x_i}\) into a disjunction of \(\SRGX\) \(\varphi_{x_i, 1}, \ldots,
\varphi_{x_i, m_i}\) by the same procedure from the proof of
Theorem~\ref{lemma:utlstar-rgx}.

After this, we generate a new set of rules, where each of this rules consist of
a possible combination of extraction expressions made by taking exactly one
disjunct \(\varphi_{x_i, j_i}\) for each variable \(x_i\). More formally, we
generate the set of rules \(R = \{\varphi_{0, k_0} \wedge x_1.\varphi_{1, k_1}
\wedge \cdots \wedge x_n.\varphi_{n, k_n} \mid (k_0, \ldots, k_n) \in [1, m_0]
\times \cdots \times [1, m_n]\}\).

Given a rule \(\alpha = \alpha_{x_0} \wedge x_1.\alpha_{x_1} \wedge \cdots
\wedge x_n.\alpha_{x_n}\) in \(R\), we can now easily transform it into a
tree-like rule. Consider, as before, any pair of nodes \(x\) and \(y\) such that
there are at least two distinct paths \(u_1, \ldots, u_{l_1}\) and \(v_1,
\ldots, v_{l_2}\) where \(u_1 = v_1 = x\) and \(u_{l_1} = v_{l_2} = y\) (the
proof can be generalized to more paths easily). Consider, without loss of
generality, that \(u_2\) appears to the left of \(v_2\) in \(\varphi_x\). Then,
for \(\alpha\) to be satisfiable, everything between \(u_2\) and \(v_2\) in
\(\varphi\) must be forced to be \(\epsilon\). Likewise, everything to the right
of \(u_3\) in \(\varphi_{u_2}\) and everything to the left of \(v_3\) in
\(\varphi_{v_2}\) must be forced to be \(\epsilon\), and so on. This can be done
in polynomial-time because it is equivalent to checking if a regular expression
accepts the word \(\epsilon\) and checking if certain variables are painted
green. As we do this, we rewrite the \(\SRGX\), removing everything but the
variables from the parts that can be \(\epsilon\). If at any point we find
an expression that cannot be empty, we remove \(\alpha\) from \(R\). Finally, we
remove every occurrence of variable \(y\) in \(\varphi_{v_{l_2 - 1}}\), thus
removing the edge from \(v_{l_2 -1}\) to \(y\) in \(G_{\alpha}\) and dissolving
the undirected cycle.

\begin{figure}
  \centering
  \caption{Undirected cycle in dag-like rule.}
  \label{fig:dag-cycle}
  \begin{tikzpicture}[
    node distance=6cm,
    nodes={draw, circle},
    font=\tiny,
    inner sep=1pt,
    minimum size=.8cm,
    mark/.style={fill=black!10},
    mark/.default=,
  ]
  \graph [math nodes] {
    x -> {
      [fresh nodes]
      u_2 -> u_3 -> "\cdots"[draw=none] -> u_{l_1 - 1},
      v_2 -> v_3 -> "\cdots"[draw=none] -> v_{l_2 - 1}
    } -> y [mark]
  };
\end{tikzpicture}
\end{figure}
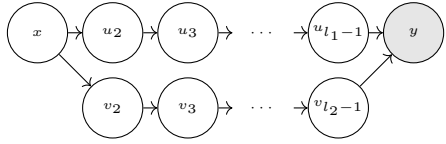

For example, consider the following dag-like rule:
\begin{equation*}
  (x \cdot \Sigma^* \cdot y) \wedge x.(a \cdot z \cdot b^*) \wedge y.(b^* \cdot z \cdot a) \wedge z.(\Sigma^*)
\end{equation*}
This rule is satisfiable only by the document \(d = aa\) and the mapping \(\mu\)
such that \(\mu(x) = (1, 2)\), \(\mu(y) = (2, 3)\), and \(\mu(z) = (2, 2)\). By
applying the procedure we described, we obtain the following rule:
\begin{equation*}
  (x \cdot y) \wedge x.(a \cdot z) \wedge y.(d) \wedge z.(\epsilon)
\end{equation*}
It is simple to observe that this rule is equivalent and tree-like.

Given the definitions of the semantics of extraction rules and \(\SRGX\), it can
be proved without difficulty that the final set of tree-like rules will be
equivalent to the initial dag-like rule. Furthermore, it is simple to see that
the final expression will be of double-exponential size with respect to the
initial dag-like rule: it will experience one exponential blow-up when the
\(\SRGX\) are transformed into disjunctions of path \(\SRGX\), and another
exponential blow-up when we generate a rule for each possible combination of
disjuncts.
\hfill\qed\medskip

\subsection*{Proof of Theorem~\ref{thm:rgx:simple}}


By Proposition~\ref{prop:nonfunc} we know that simple rules are equivalent to
unions of functional dag-like rules, by Proposition~\ref{thm:dag_to_tree} we
know that satisfiable dag-like rules are equivalent to unions of functional
tree-like rules. Finally, by Lemma~\ref{lemma:utlstar-rgx} we know that unions
of functional tree-like rules are equivalent to \(\RGX\).
\hfill\qed\medskip


\section{PROOFS FROM SECTION~\secref{sec:evaluation}}
\medskip

\subsection*{Proof of Theorem~\ref{prop:evaluation}}


\begin{algorithm}[t]
  \caption{Enumerate all spans in $\semd{\gamma}$. Here, $\gamma$ is the
    expression being evaluated, $d$ is the document, $\mu$ is the current
    mapping, and $V$ is the set of available variables.}
  \label{alg:enum}
  \begin{algorithmic}[1]
    \Procedure{Enumerate}{$\gamma, d, \mu, V$}
      \If{\(\EVAL[\LL](\gamma, d, \mu)\) is false} \label{alg:enum:eval-check}
        \Return
      \EndIf
      \If{\(V = \emptyset\)} \label{alg:enum:finish-check}
        \textbf{output} \(\mu\) and \Return
      \EndIf
      \State{\textbf{let} \(x\) be some element from \(V\)} \label{alg:enum:select-var}
      \For{\(s \in \sub(d) \cup \{\bot\}\)} \label{alg:enum:loop-span}
        \State{\Call{Enumerate}{$\gamma, d, \mu[x \to s], V\setminus\{x\}$}}
      \EndFor
    \EndProcedure
  \end{algorithmic}
\end{algorithm}

The algorithm for enumerating all mappings for an expression \(\gamma\) on a
document \(d\) is described in Algorithm~\ref{alg:enum}. For enumerating all
mappings, one would have to call \textsc{Enumerate}\((\gamma, d, \emptyset,
\Vars(\gamma))\). We denote as ``\textbf{output}'' the operation of outputting a
mapping and then continuing computation from that point. When \(\Vars(\gamma)\)
is empty, then we simply return the empty mapping \(\emptyset\) if
\(\EVAL[\LL](\gamma, d, \emptyset)\), and nothing otherwise.

It is easy to observe that \(\mu \in \semd{\gamma}\) if and only if
\(\EVAL[\LL](\gamma, d, \mu_{\bot})\), where \(\mu_{\bot}(x)\) is \(\mu(x)\) if
\(x \in \dom(\mu)\), and \(\bot\) otherwise. It is also easy to observe that for
every mapping \(\mu\), it holds that \(\emptyset \subseteq \mu\). From these two
observations, and given a particular \(\mu \in \semd{\gamma}\), it is
straightforward to prove by induction that the algorithm will eventually output
\(\mu\).

Finally, we prove that this is a \textit{polynomial delay}
algorithm. Notice that the algorithm will only recurse if there exists a mapping
\(\mu'\) such that \(\mu' \in \semd{\gamma}\) and \(\mu[x \to s] \subseteq
\mu'\). Since \(|\sub(d) \cup \{\bot\}| \leq |d|^{2} + 1\) and the algorithm can
only recurse up to a depth of \(|\VV|\), the function \(\EVAL[\LL]\) will be
called at most \(|\VV|(|d|^{2} + 1)\) times before an output is reached (or the
algorithm terminates). Given that \(\EVAL[\LL]\) can be decided in polynomial
time, the time to produce the next output will be polynomial.
\hfill\qed\medskip

\subsection*{Proof of Theorem~\ref{theo:eval-hardness-general}}


\newcommand{\mcap}{\operatorname{cap}}
\newcommand{\mcomp}{\operatorname{comp}}
\newcommand{\mtemp}{\operatorname{eq}}
\newcommand{\mfalse}{\operatorname{conflict}}

To prove that $\nemp[\SRGX]$ is $\NP$-hard, we provide a reduction from
1-IN-3-SAT. The input of 1-IN-3-SAT is a propositional formula $\alpha = C_1
\wedge \cdots \wedge C_n$, where each $C_i$ ($1 \leq i \leq n$) is a disjunction
of exactly three propositional variables (negative literals are not allowed).
Then the problem is to verify whether there exists a satisfying assignment for
$\alpha$ that makes exactly one variable per clause true. 1-IN-3-SAT is known to
be $\NP$-complete (see \cite{GJ79}).

For the reduction, we construct a $\SRGX$ $\gamma_\alpha$ such that
$\semg{\gamma_\alpha}$ is not empty if and only if there exists a satisfying
assignment for $\alpha$ that makes exactly one variable per clause true, with $d
= \epsilon$. In this reduction, we assume that for every clause $C_i$ in
$\alpha$ ($1 \leq i \leq n$), it holds that $C_i = (p_{i,1} \vee p_{i,2} \vee
p_{i,3})$, where each $p_{i,j}$ is a propositional variable. Notice that
distinct clauses can have propositional variables in common, which means that
$p_{i,j}$ can be equal to $p_{k, \ell}$ for $i \neq k$.

To define $\gamma_\alpha$ we consider two sets of variables: $\{ x_{i,j} \mid 1
\leq i \leq n$ and $1 \leq j \leq 3\}$ and $\{ y_{i,j,k,\ell} \mid 1 \leq i < k
\leq n$ and $1 \leq j,\ell \leq 3\}$. With these variables we encode the truth
values assigned to the propositional variables in $\alpha$; in particular, a
span is assigned to the variable $x_{i,j}$ if and only if the propositional
variable $p_{i,j}$ is assigned value true. Moreover, $\gamma_\alpha$ is used to
indicate that exactly one of $p_{i,1}$, $p_{i,2}$ and $p_{i,3}$ is assigned
value true, which is essentially represented by a $\SRGX$ of the form $(x_{i,1}
\vee x_{i,2} \vee x_{i,3})$, indicating that exactly one of $x_{i,1}$, $x_{i,2}$
and $x_{i,3}$ has to be assigned a span. Besides, $\gamma_\alpha$ is used to
indicate that if $p_{i,j}$ is assigned value true, then we are forced to assign
value false not only to $p_{i,k}$ with $k \neq j$ but also to some propositional
variables in other clauses. This idea is formalized by means of the notion of
conflict between propositional variables. More precisely, we say that $p_{i,j}$
is in conflict with $p_{k,\ell}$ if $i < k$ and one of the following conditions
holds:
\begin{itemize}
\item there exists $m \in \{1,2,3\}$ such that $p_{i,j} = p_{k,m}$ and $m \neq \ell$;
\item there exists $m \in \{1,2,3\}$ such that $p_{i,m} = p_{k,\ell}$ and $m \neq j$.
\end{itemize}
Thus, if $p_{i,j}$ is assigned value true and $p_{i,j}$ is in conflict with
$p_{k,\ell}$, then we know that $p_{k,\ell}$ has to be assigned value false. In
$\gamma_\alpha$, the variable $y_{i,j,k,\ell}$ is used to indicate the presence
of such a conflict; in particular, a span is assigned to $y_{i,j,k,\ell}$ if and
only if the propositional variable $p_{i,j}$ is in conflict with the
propositional variable $p_{k,\ell}$. We collect all the conflicts of $p_{i,j}$
in the set $\mfalse(p_{i,j})$:
\begin{align*}
\{ y_{i,j,k,\ell} \mid p_{i,j} \text{ is in conflict with } p_{k,\ell}\} \ \cup
\ \{ y_{k,\ell,i,j} \mid p_{k,\ell} \text{ is in conflict with } p_{i,j}\}
\end{align*}
The variable $y_{i,j,k,\ell}$ is used as follows in $\gamma_\alpha$. If some
spans have been assigned to $x_{i,j}$ and $y_{i,j,k,\ell}$, then no span is
assigned to $x_{k,\ell}$, as the propositional variable $p_{i,j}$ has been
assigned value true and $p_{i,j}$ is in conflict with the propositional variable
$p_{k,\ell}$. To encode this restriction, define the $\SRGX$ $\gamma_{i,j}$ as
the concatenation of the variables in $\mfalse(p_{i,j})$ in no particular order.
For example, if
\begin{eqnarray*}
\mfalse(p_{3,1}) &=& \{y_{1,2,3,1}, y_{1,3,3,1}, y_{3,1,4,1}, y_{3,1,5,2}\},
\end{eqnarray*}
then
\begin{eqnarray*}
\gamma_{3,1} & = & y_{1,2,3,1} \cdot y_{1,3,3,1} \cdot y_{3,1,4,1} \cdot y_{3,1,5,2}
\end{eqnarray*}
Finally, for every clause $C_i$ ($1\leq i \leq n$) define $\SRGX$ $\gamma_i$
as:
\begin{align*}
&(x_{i,1} \cdot \gamma_{i,1} \ \vee \ x_{i,2} \cdot \gamma_{i,2} \ \vee \ x_{i,3} \cdot \gamma_{i,3})
\end{align*}
With this notation, we define \(\SRGX\) $\gamma_\alpha$ as follows:
\begin{eqnarray*}
\gamma_\alpha & =  & \gamma_1 \cdots \gamma_n
\end{eqnarray*}
At this point it is important to understand how the variables $y_{i,j,k,\ell}$
are used in the $\SRGX$ $\gamma_\alpha$. Assume that $p_{1,1} = p_{2,1}$, so
that $p_{1,1}$ is in conflict with $p_{2,2}$. Then if we assigned value true to
$p_{1,1}$, we have that $p_{2,1}$ is also assigned value true, so $p_{2,2}$ has
to be assigned value false. This restriction is encoded by using the variable
$y_{1,1,2,2}$. More precisely, $\gamma_\alpha = \gamma_1 \cdot \gamma_2 \cdots
\gamma_n$, where $\gamma_1$ is of the form:
\begin{align*}
  (x_{1,1} \cdots y_{1,1,2,2} \cdots \ \vee\ x_{1,2} \cdot \gamma_{1,2} \ \vee \  x_{1,3} \cdot \gamma_{1,3}),
\end{align*}
given that $y_{1,1,2,2} \in \mfalse(p_{1,1})$, and $\gamma_2$ is of the form:
\begin{align*}
  (x_{2,1} \cdot \gamma_{2,1} \ \vee\ x_{2,2} \cdots y_{1,1,2,2} \cdots \ \vee \  x_{2,3} \cdot \gamma_{2,3}),
\end{align*}
given that $y_{1,1,2,2} \in \mfalse(p_{2,2})$. Thus, if $x_{1,1}$ is assigned a
span, representing the assignment of value true to the propositional variable
$p_{1,1}$, then also $y_{1,1,2,2}$ is assigned a span (both spans will have
empty content by the definition of $\gamma_\alpha$ and $d$). If we now try to
assign a span to $x_{2,2}$, then we are forced to assign a span to
$y_{1,1,2,2}$ again. This, however, violates the definition of the semantics of
$\RGX$, because the mappings for concatenated expressions must have disjoint
domains (in other words, they cannot both assign the same variable).

Based on the previous intuition, it is straightforward to prove that
$\semg{\gamma_\alpha}$ is not empty if and only if there exists a satisfying
assignment for $\alpha$ that makes exactly one variable per clause true, which
was to be shown. As before, we take $d$ to be $\epsilon$.
\hfill\qed\medskip

\subsection*{Proof of Proposition~\ref{theo:ptime-functional-RGX}}


Since \(\fRGX\) is a subset of sequential \(\RGX\), this is implied by
Theorem~\ref{theo:ptime-sequential-sVA}.
\hfill\qed\medskip

\subsection*{Proof of Proposition~\ref{theo:eval-hardness-relational-VA}}


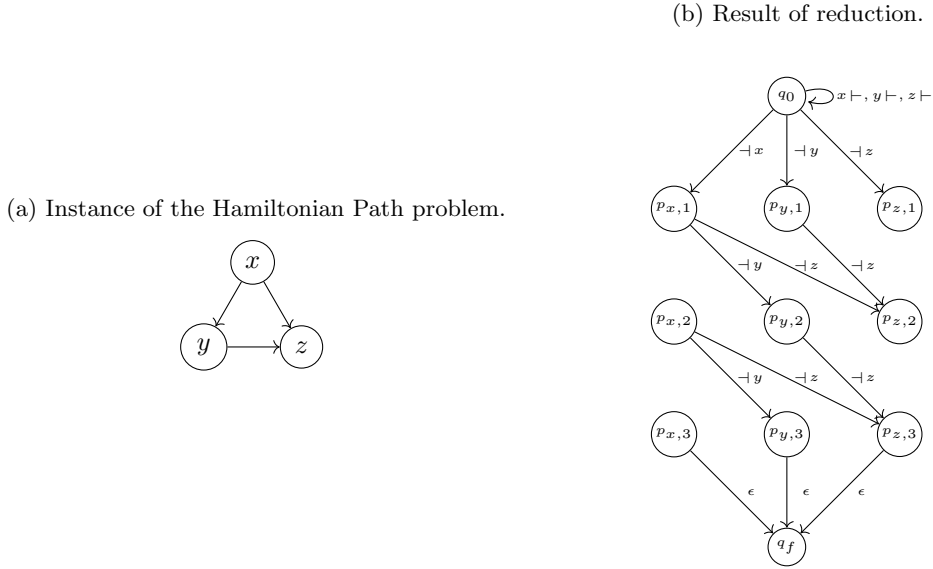
\begin{figure}
  \centering
  \caption{Reduction of Proposition~\ref{theo:eval-hardness-relational-VA}}
  \label{fig:hamiltonian-path}
  \begin{subfigure}[c]{0.4\textwidth}
    \caption{Instance of the Hamiltonian Path problem.}
    \centering
    \begin{tikzpicture}[
    node distance=.2cm,
    nodes={draw, circle},
  ]
  \graph [math nodes, clockwise, n=3, radius=.75cm, ->] {
    x, z, y
  };
  \path[->]
  (x) edge (y)
      edge (z)
  (y) edge (z)
  ;
\end{tikzpicture}
  \end{subfigure}
  \begin{subfigure}[c]{0.4\textwidth}
    \caption{Result of reduction.}
    \centering
    \makeatletter
\tikzset{%
  show node name/.code={%
    \expandafter\def\expandafter\tikz@mode\expandafter{\tikz@mode\show\tikz@fig@name}%
    }
}
\makeatother
\begin{tikzpicture}[
    node distance=6cm,
    nodes={draw, circle},
    font=\tiny,
    inner sep=1pt,
    minimum size=.5cm,
    mark/.style={fill=black!10},
    mark/.default=,
  ]
  \foreach \var/\i in {x/1,y/2,z/3} {
    \foreach \j in {1,2,3} {
      \node (\var\j) at (1.5*\i, -1.5*\j) {$p_{\var,\j}$};
    }
  }
  \node (q0) at (3, 0) {$q_0$};
  \node (qf) at (3, -6) {$q_f$};
  \path[->, nodes={draw=none, right}]
  (q0) edge [loop right] node {$\Open{x},\Open{y},\Open{z}$} ()
  \foreach \var in {x,y,z} {
    (q0) edge node {$\Close{\var}$} (\var1)
    (\var 3) edge node {$\epsilon$} (qf)
  }
  \foreach \from/\to in {x/y,x/z,y/z} {
    (\from 1) edge node {$\Close{\to}$} (\to 2)
    (\from 2) edge node {$\Close{\to}$} (\to 3)
  }
  ;
\end{tikzpicture}
  \end{subfigure}
\end{figure}

We will prove that \(\nemp\) of relational \(\VASet\) automata is
\(\NP\)-complete and we will also prove that the \(\mdcheck\) problem is
\(\NP\)-complete. The \(\mdcheck\) problem receives as input an expression
\(\gamma\), a document \(d\), and a mapping \(\mu\), and asks whether \(\mu \in
\semd{\gamma}\).

The membership of both problems to \(\NP\) is very easy to prove. In both cases
we only need to guess a run for the variable automaton (that conforms to the
input document and mapping) and verify that it is accepting. The size of the
runs that we need to consider is bounded by a polynomial because the document
and the available variables are part of the input. Furthermore, we only need to
consider sequences of consecutive \(\epsilon\)-transitions that are shorter than
the number of states in the variable automaton. This is because longer sequences
will inevitably have a cycle, which can be removed without altering the
acceptance of the run.

To prove \(\NP\)-hardness of the \(\mdcheck\) problem we will describe a
reduction from the \textit{Hamiltonian path problem}. This problem consists in
deciding whether or not a directed graph has a path that visits every vertex
exactly once, and it is known to be \(\NP\)-hard (\cite{GJ79}). Let \(G = (V,
E)\) be a graph and let \(A = (Q, q_0, q_f, \delta)\) be the variable automaton
that results from reducing \(G\). We will construct \(A\) in such a way that
\(G\) has a Hamiltonian path if and only if \(\mu_{\epsilon} \in \semd{A}\),
where \(d = \epsilon\) and \(\mu_{\epsilon}\) is such that \(\mu_{\epsilon}(x) =
(1, 1)\) for all \(x \in \Var(A)\).

The automaton \(A\) is built as follows: (1) for every vertex \(v \in V\), add
states \(p_{v,1}, p_{v,2}, \ldots, p_{v,|V|}\) to \(Q\); (2) for every edge
\((u, v)\) and every \(i \in [1, |V| - 1]\) add the transitions \((p_{u,i},
\Close{x_v}, p_{v,i+1}), (q_0, \Close{x_v}, p_{v,1})\) to \(\delta\); (3) add
two fresh states for \(q_0, q_f\) and, for every \(v \in V\) add transitions
\((p_{v,|V|}, \epsilon, q_f), (q_0, \Open{x_v}, q_0)\) to \(\delta\).
Figure~\ref{fig:hamiltonian-path} shows an example of this reduction. Notice
that every accepting run of \(A\) assigns every variable to the span \((1, 1)\),
since to go from \(q_0\) to \(q_f\) it must go through \(|V|\) closing
transitions (which must be different if the run is valid). Thus, \(A\) is
relational. Because the states and transitions in \(A\) correspond to the
vertices and edges in \(G\) there will be a one-to-one correspondence between
runs in \(A\) and Hamiltonian paths in \(G\). That is, if there is an accepting
run that goes through states \(p_{v_1,1}, \ldots, p_{v_{|V|},|V|}\), then there
is a Hamiltonian path through the vertices of \(G\). Proving this last statement
is straightforward given the way \(A\) was built.

To see why the \(\nemp\) is also \(\NP\)-hard, notice that in the aforementioned
construction, when graph \(G\) does not have a Hamiltonian path there will be no
accepting runs. Therefore, it holds that \(\semd{A}\) is not empty if and only
if \(G\) has a Hamiltonian path.
\hfill\qed\medskip

\subsection*{Proof of Proposition~\ref{theo:ptime-sequential-check}}


We describe an algorithm for checking if a variable automaton is sequential that
is in \(\coNLOGSPACE\), which is known to be equal to \(\NLOGSPACE\)
(\cite{Immerman88}). 

The algorithm will non-deterministically traverse the automaton searching for a
non-sequential path. To do so, it remembers the current variable's status, which
can be either \textit{available}, \textit{open} or \textit{closed}. If it finds
a transition which is incompatible with the current status (e.g.\ opening an
already open variable), it accepts, indicating that the variable automaton is
not sequential. More formally, let \(A = (Q, q_0, q_f, \delta)\) be a variable
automaton, let \(q_{curr}\) denote the current state and let \(s_{curr}\) denote
the current variable's status. We also use the value $o$ to count the number of opened variables.
For every variable \(x \in \VV\) the algorithm
proceeds as follows:
\begin{itemize}
\item Set \(q_{curr}\) to \(q_0\) and \(s_{curr}\) to \textit{available}; set $o$ to 0;
\item while \(q_{curr} \neq q_f\):
  \begin{itemize}
  \item non-deterministically pick a transition \((q_{curr}, a, q_{next}) \in
    \delta\);
  \item if \(a\) is incompatible with the status \(s_{curr}\), then accept;
    otherwise, update \(q_{curr}\) to \(q_{next}\) and \(s_{curr}\) according to
    \(a\); also increase or decrease the value of $o$ as stipulated by $a$; 
  \end{itemize}
\item if $o\neq 0$ accept;
\item otherwise reject.
\end{itemize}

It is simple to realize that the algorithm is correct, since if there is a
non-sequential path, then there is a sequence of non-deterministic decisions
that will lead the algorithm to accept. On the other hand, it is also apparent
that the algorithm uses only logarithmic space, because it only has to store the
current variable, current state, next state, variable status and the number of variables in $o$; in other words,
a constant amount of information that is at most logarithmic in size with
respect to the input. \hfill\qed\medskip

\subsection*{Proof of Proposition~\ref{prop:rgx-equiv-rgx-seq}}


To prove this we use the \textit{path union stack variable automata}
construction detailed in \cite[Subsection~4.1.2]{FKRV15}, that can be easily
adapted to the definition of \(\RGX\) proposed in this work.

Specifically, let \(\gamma\) be a \(\RGX\). By Theorem~\ref{th-stack_rgx}, we
know that \(\gamma\) has an equivalent \(\VAStk\) automaton \(A\). By the result
in \cite[Lemma~4.3]{FKRV15}, we know that \(A\) has an equivalent path union
\(\VAStk\) (denoted \(\PUStk\)) \(A'\). Given the construction of \(A'\), it is
easy to observe that every path in it will be sequential, which implies that
\(A'\) as a whole is sequential. We may build a \(\RGX\) \(\gamma'\) equivalent
to \(A'\) by transforming each path in \(A'\) into a sequential \(\RGX\), and
then joining the resulting expressions with disjunctions. It is clear that
disjunctions of sequential \(\RGX\) are sequential. Therefore, \(\gamma'\) is
sequential and equivalent to \(\gamma\). \hfill\qed\medskip

\subsection*{Proof of Theorem~\ref{theo:ptime-sequential-sVA}}


We will reduce the \(\EVAL\) problem in sequential \(\RGX\) to the same problem
on sequential variable automata, and show that the latter can be decided in
\(\PTIME\). Let \(\gamma\) be a sequential \(\RGX\). We can adapt the Thompson
construction algorithm \cite{HopcroftU79} to transform \(\gamma\) into a
variable automaton in polynomial time. We now prove by induction that \(A\) will
be sequential, given the fact that \(\gamma\) is sequential. We need to consider
the following cases:
\begin{itemize}
\item \(\gamma = a\), where \(a \in \Sigma\): this is the base case and the
  automaton is trivially sequential.
\item \(\gamma = \psi_1 \cdot \psi_2\): it is very easy to observe that the
  concatenation of two sequential paths that use disjoint sets of variables, is
  sequential.
\item \(\gamma = \psi_1 \vee \psi_2\): every path will be in either the
  automaton for \(\psi_1\) or the automaton for \(\psi_2\), which are sequential
  by the inductive hypothesis.
\item \(\gamma = (\psi)^{*}\): the set of variables in \(\psi\) is empty, its
  automaton is, thus, trivially sequential.
\item \(\gamma = x\{\psi\}\), where \(x \in \VV\): since \(\psi\) does not use
  \(x\), it is trivial to prove that every path  will be sequential.
\end{itemize}
Therefore, automata constructed from sequential expressions will be sequential.

Next, we prove that the \(\EVAL\) problem for sequential variable automata is in
\(\PTIME\). The main idea behind this proof, and many of the following, will be
to embed in document \(d\) the variable operations corresponding to mapping
\(\mu\). This will allow us to then to treat variable operation transitions as
normal transitions. This is an advantage because then we can use classical
algorithms for finite automata to decide problems.

Let \(\Op(A) = \{\Open{x}, \Close{x} \mid x \in \Vars(A)\}\). Let \(\rho\) be a
run for document \(d\) and mapping \(\mu\) on a variable automaton \(A\). We
refer to the \textit{label} of \(\rho\), denoted \(L(\rho)\), as the string
\(\lambda \in (\Sigma \cup \Op(A))\) that is the concatenation of the labels of
the transitions in \(\rho\), in the order they are used.

Given a label \(\lambda\), we may easily generate the document-mapping pair
\((d, \mu)\) from the run of \(\lambda\) in logarithmic-space. We simply scan
\(\lambda\) from left to right, outputting symbols of \(\Sigma\) to \(d\), then
we do a second scan, counting symbols to determine the spans of \(\mu\). It is
simple to see that if we change the order of consecutive variable operation in
\(\lambda\), then the generated \((d, \mu)\) will be the same.

As an example, consider the document \(d = abc\) and the mapping \(\mu\) such
that \(\mu(x) = (1,3)\) and \(\mu(y) = (3,3)\). Some labels that correspond to
these are \(\lambda_1 = \Open{x},a,b,\Open{y},\Close{x},\Close{y},c\) or
\(\lambda_2 = \Open{x},a,b,\Close{x},\Open{y},\Close{y},c\).

Similarly, for every pair \((d, \mu)\), and a fixed set of variables, there is a
finite set of possible labels of runs that correspond to \(d\) and \(\mu\). By
the previous paragraph, it is an easy observation that the labels in this set
will differ only on the ordering of consecutive variable operations, and
variables that are opened but never closed.

Since the ordering of variable operations will be problem in most proofs, we
will frequently use the technique of \textit{coalescing} consecutive variable
operations. What this means, is that we will consider a set of consecutive
variable operations as a single symbol. We will usually accompany this by
introducing new transitions to the automata that recognize these coalesced
symbols.

Let \(A = (Q, q_0, q_f, \delta)\) be the sequential automaton, \(d\) the
document and \(\mu\) the mapping. First, let \(\lambda\) be some label for \((d,
\mu)\). Let \(\tau = T_1, \ldots, T_\ell\) be a partition of \(\dom(\mu)\) such
that two variable operations \(o_1\) and \(o_2\) belong to the same \(T_i\) if
and only if \(o_1 \cdot w \cdot o_2\) is a substring of \(\lambda\) and \(w\) is
\(\epsilon\) or consists solely of variable operations. We treat the sets in
\(\tau\) as new symbols of the alphabet. We will coalesce all sequences of
consecutive variable operations in \(\lambda\) replacing them with their
respective \(T_i\), and call the result \(d'\).

Let \(A' = (Q, q_0, q_f, \delta')\) be as follows. For each transition \((p, a,
q) \in \delta\): (1) if \(a \in \Sigma \cup \{\epsilon\}\), then \((p, a, q) \in
\delta'\); (2) if \(a\) is a variable operation for \(x\) and \(x \not\in
\dom(\mu)\), then \((p, \epsilon, q) \in \delta'\); otherwise, ignore the
transition. Finally, for every set \(T_i\) (\(i \in [1, \ell]\)), transition
\((p, T_i, q) \in \delta'\) if there exists a path from \(p\) to \(q\) in \(A\)
satisfying the following conditions: (1) every transition in the path is either
an \(\epsilon\)-transition or corresponds to a variable operation in \(T_i\);
and (2) for every variable operation in \(T_i\), there is exactly one transition
in the path that corresponds to it. Notice that \(A'\) has no variable
operations, and therefore, behaves exactly like a non-deterministic finite
automaton. Therefore, the problem has been reduced to that of deciding whether
the non-deterministic finite automaton \(A'\) accepts the word \(d'\), which is
known to be in \(\PTIME\) (\cite{HopcroftU79}).

Except for the last step, it is clear that this reduction runs in polynomial
time. Therefore, in order to complete this part of the proof, we only need to
provide an algorithm that given states \(p, q\) and \(i \in [1, n + 1]\) decides
whether \((p, T_i, q) \in \delta'\). We will describe an algorithm that finds a
path in \(A\) that in \(\NLOGSPACE\), which is contained in \(\PTIME\)
(\cite{papadimitriou1993computational}). Taking into account that \(A\) is
sequential, we know that the paths will not repeat operations nor execute them
in a wrong order, therefore, we only need to count the number of variable
operations.

The algorithm starts from state \(p\) and sets a counter \(c\) to 0. Then, at
each step it guesses the next transition, and checks that it is either an
\(\epsilon\)-transition or corresponds to a variable transition in \(T_i\). If
it is the latter, then it increments \(c\) by one. If the algorithm reaches
\(q\), it accepts only if \(c = |T_i|\). From the description of the algorithm
it is straightforward to prove that it is correct and uses logarithmic-space.

Now we prove the correctness of the algorithm. Namely, we will prove that there
exists an extension \(\mu'\) of \(\mu\) if and only if \(A'\) accepts \(d'\). We
will consider the three cases that can happen to a variable \(x\) with respect
to \(\mu\): (1) \(x \not\in \dom(\mu)\), (2) \(\mu(x) = \bot\), and (3) \(\mu(x)
= (i, j)\) for \(i, j \in [1, n + 1]\). In case (1), we have that \(x\) may or
may not be in \(\dom(\mu')\). This agrees with the fact that variable operations
for \(x\) are replaced with \(\epsilon\) in \(A'\). Furthermore, because \(A\)
is sequential, we know that there are no valid runs in \(A'\) that would be
invalid in \(A\). In case (2), \(\mu'\) cannot assign \(x\), which agrees with
\(A'\) because variable operations for \(x\) were removed. Finally, in case (3),
we know that \(mu'\) will be compatible with \(\mu\) on \(x\) because each of
the \(T_i\) symbols we introduced can be matched by \(A'\) if and only if there
exists a path in \(A\) that performs the variable operations in \(T_i\) in some
order. Given these observations it is very apparent that there is a one-to-one
correspondence between accepting runs in \(A\) and \(A'\), which finishes the
proof of correctness.
\hfill\qed\medskip

\subsection*{Proof of Theorem~\ref{theo:eval-daglike-npcomplete}}


First, we show that the problem is in \(\NP\). Consider a rule \(\varphi\) that
uses functional \(\SRGX\), and a document \(d\). To decide the problem we can
guess a mapping \(\mu\), which is of polynomial size, and we check that \(\mu
\in \semd{\varphi}\). This can be done in polynomial time for the following
reason. From Theorem~\ref{theo:ptime-sequential-sVA}, we know that \(\EVAL\)
of sequential (and thus functional) \(\RGX\) is in \(\PTIME\). Therefore, we can
easily check that \(\mu\) respects the semantics of rules (with regards to
instantiated variables, for example) and for each relevant extraction expression
\(x.\varphi_x\), we can check that \(\mu\) restricted to \(\Var(\varphi_x)\)
satisfies \(\varphi_x\) when \(d\) is restricted to \(\mu(x)\).

To show that the problem is \(\NP\)-hard, we will describe a polynomial time
reduction from the \(\ONETHREESAT\) problem. The input for \(\ONETHREESAT\)
consists of a propositional formula \(\alpha = C_1 \wedge \cdots \wedge C_n\)
where each clause \(C_i\) (\(1 \leq i \leq n\)) is a disjunction of three
positive literals: \(p_{i,1}\), \(p_{i,2}\), and \(p_{i,3}\). The problem is to
determine if there is a truth assignment that makes \emph{exactly} one literal
true in each clause. This problem is known to be \(\NP\)-complete (\cite{GJ79}).

Given the propositional formula \(\alpha\), the reduction will output a rule
\(\varphi\) using functional \(\SRGX\) and a document \(d\) such that
\(\semd{\varphi}\) is non-empty if and only if \(\alpha\) is satisfiable. Let
\(V\) be the set of variables in \(\alpha\). In \(\varphi\) we use the variables
in \(V\) plus fresh variables \(c_i\) for \(i \in [1, n]\) and two extra
variables: \(T\) and \(F\). The rule \(\varphi\) consists of the following
extraction expressions:
\begin{itemize}
  \item \(T \cdot c_1 \cdot F\);
  \item \(c_i.(p_{i,1} \cdot c_{i + 1} \cdot p_{i,2} \cdot p_{i,3}) \vee
    (p_{i,2} \cdot c_{i + 1} \cdot p_{i,1} \cdot p_{i,3}) \vee (p_{i,3} \cdot
    c_{i + 1} \cdot p_{i,1} \cdot p_{i,2})\) for \(i \in [1, n - 1]\); and
  \item \(c_n.(p_{i,1} \cdot T \cdot \# \cdot F \cdot p_{i,2} \cdot p_{i,3})
    \vee (p_{i,2} \cdot T \cdot \# \cdot F \cdot p_{i,1} \cdot p_{i,3}) \vee
    (p_{i,3} \cdot T \cdot \# \cdot F \cdot p_{i,1} \cdot p_{i,2})\), where
    \(\#\) is a symbol in the alphabet.
\end{itemize}
Note that every \(\SRGX\) is functional.

The intuition behind the reduction is that every variable placed to the left of
the \(\#\) symbol would be assigned a true value, and every variable placed to
the right of the symbol would be assigned a false value. Notice that \(\varphi\)
can only be satisfied by the document \(d = \#\) and a mapping \(\mu\) such that
\(\mu(T) = (1, 1)\) and \(\mu(F) = (2,2)\). If \(\mu\) satisfies \(\varphi\), we
can make the following observations: (1) for every \(x \in V\), either \(\mu(x)
= (1,1)\) or \(\mu(x) = (2,2)\); and (2) for every \(i \in [1, n]\), there is
exactly one \(j \in \{1, 2, 3\}\) such that \(\mu(p_{i,j}) = (1,1)\). With these
observations in mind, it is easy to see that every satisfying mapping of
\(\varphi\) will correspond to a satisfying truth assignment of \(\alpha\) and
vice versa, thus proving the reduction correct.
\hfill\qed\medskip

\subsection*{Proof of Theorem~\ref{theo:treelike-ptime}}


In order to prove that \(\EVAL\) of sequential tree-like rules is in \(\PTIME\),
we will describe an algorithm that first does some polynomial-time preprocessing
of the input, and then runs in \textit{alternating logarithmic space}
(\(\ALOGSPACE\)), which is known to be equivalent to \(\PTIME\)
(\cite{papadimitriou1993computational}).

Let \(\varphi = \varphi_{x_0} \wedge x_1.\varphi_{x_1} \wedge \cdots \wedge
x_m.\varphi_{x_m}\) be a sequential tree-like rule with graph \(G_{\varphi}\),
let \(d\) be a document, and let \(\mu\) be a mapping. We assume, without loss
of generality, that for every variable \(x\) in \(\varphi\) there is an
extraction expression \(x.\varphi_x\) in \(\varphi\).

We may immediately reject in two cases: (1) \(\mu\) is not hierarchical; and (2)
there are variables \(x\) and \(y\) such that \(\mu(x) = \mu(y)\), the content
of \(\mu(x)\) is not empty, and there is no directed path in \(G_{\varphi}\)
that connects \(x\) and \(y\). These two cases can easily be checked in
polynomial-time, and will help us simplify the proceeding analysis.

For the purpose of this proof, we say that two variables \(x\) and \(y\) are
\textit{indistinguishable} if \(\mu(x) = \mu(y) = (i, i)\) for some \(i \in [1,
n + 1]\) and they are siblings in \(G_{\varphi}\); that is, there exists a
variable \(z\) such that \((z, x)\) and \((z, y)\) are edges in \(G_{\varphi}\).
The problem with these variables is that we cannot deduce from \(\mu\) and
\(\varphi\) the order in which they must be encountered when processing the
document. Therefore, we will coalesce each set of indistinguishable variables
into a single variable. This means removing these variables from the global set
of variables and replacing them with a single new variable that represents the
set. We refer to these new variables as \textit{coalesced variables}, and we
refer to mapping \(\mu\) updated to reflect this change as \(\mu'\).

By coalescing indistinguishable variable, however, we will be destroying the
subtrees rooted at them. Therefore, we must check that \(\mu\) agrees with this
subtrees. Let \(U\) be a maximal set of pairwise indistinguishable variables.
For each \(x \in U\) we perform the following ``emptiness'' check. Transform
\(\varphi_{x}\) into a variable automaton \(A_x\) and check that: (1) there is a
path from the initial state to the final state of \(A_x\) that uses only
\(\epsilon\)-transitions and variable operations; (2) this path opens and closes
every variable \(y\) such that \((x, y)\) is in \(G_{\varphi}\); (3) for every
variable \(y\) used in this path, either \(\mu(x) = \mu(y)\) or \(y \not\in
\dom(\mu)\); and (4) recursively perform the ``emptiness'' check on \(y\) and
\(\varphi_{y}\). This may be done in polynomial time by using similar techniques
to those shown on the proof of Theorem~\ref{theo:ptime-sequential-sVA}.

For this proof, we will use again the idea of \textit{labels} (defined in the
proof of Theorem~\ref{theo:ptime-sequential-sVA}). Notice that if we fix an
order \(\preceq_{\Op{}}\) of the variable operations and limit to those
variables in \(\dom(\mu)\), then there is a unique label for \((d, \mu)\) in
which consecutive variable operations are ordered according to
\(\preceq_{\Op{}}\). We denote this label \(L(d, \mu, \preceq_{\Op{}})\), and we
may compute it easily in polynomial-time.

In addition to the above, we say that a label \(\lambda\) is \textit{balanced}
if all of its opening and closing variable operations are correctly balanced
(like parentheses). It is clear that given a valid \((d, \mu)\), \(\mu\) is
hierarchical if and only if \((d, \mu)\) have at least one balanced label.

Now, notice that if we take into account \(\mu'\), \(G_{\varphi}\) and
indistinguishable variables, then there is a unique order in which variable
operations could be seen by the rule if the document is processed sequentially.
We will use this order as the order \(\preceq_{\Op{}}\), which we can compute as
follows. Let \(V = \{x \in \dom(\mu') \mid \mu(x) \neq \bot\}\) and consider the
induced subgraph \(T = G_{\varphi}[V]\). A node \(x\) in \(T\) precedes its
sibling \(y\) if \(\mu'(x) = (i, j)\), \(\mu'(y) = (k, l)\), and \(\min(i, j) <
\max(k, l)\). Since we coalesced indistinguishable variables, we know that there
is a unique way to put siblings in this order. Finally, the order can be
obtained by doing an ordered depth-first search on \(T\): when we enter a node
\(x\) we add \(\Open{x}\) to the output, when we finish processing the subtree
rooted at \(x\) we add \(\Close{x}\) to the output. With this in mind, we define
the document \(d' = L(d, \mu', \preceq_{\Op{}})\).

Next, we transform each sequential \(\SRGX\) \(\varphi_{x_i}\) into a
non-deterministic finite automaton \(A_{x_i} = (Q, q_0, q_f, \delta)\). For each
coalesced variable \(X\) that represents the set of indistinguishable variables
\(U\), we add a new state \(q_X\) and transitions \((p, \Open{X}, q_X)\) and
\((q_X, \Close{X}, q)\) if there is a path from \(p\) to \(q\) that uses only
\(\epsilon\)-transitions and variable transitions such that every variable in
set \(U\) is opened and closed in this path. This can be done in polynomial-time
because all expressions are sequential (the same way it was done on the proof of
Theorem~\ref{theo:ptime-sequential-sVA}).

Now, we run the alternating logarithmic space algorithm. We will have two
pointers: \(i_{curr}\) and \(i_{end}\). They will denote the part of the
document that we are considering at any given time, and will start as 1 and
\(|d'| + 1\) respectively. The algorithm works by traversing the automata
guessing transitions. Every time we choose a transition in \(A_x\) that opens
variable \(y\), we find the position \(i_{close}\) in \(d'\) where \(y\) is
closed (or guess it if \(y \not\in \dom(\mu')\)) and check two conditions in
parallel (by use of alternation): (1) \(A_y\) recursively accepts \((d', \mu')\)
on the interval \((i_{curr}, i_{close})\); and (2) \(A_x\) accepts \((d',
\mu')\) on the interval \((i_{close}, i_{end})\), continuing from the current
state. More specifically, the algorithm is the following:
\begin{enumerate}
  \item Set \(i_{curr}\) to 1, \(i_{end}\) to \(|d'| + 1\), and \(x_{curr}\) to \(x_0\).
  \item Let \(A_{x_{curr}}\) be \((Q, q_0, q_f, \delta)\).
  \item Set \(q_{curr}\) to \(q_0\).
  \item While \(q_{curr} \neq q_f\) and \(i_{curr} \leq i_{end}\):
    \begin{enumerate}
    \item Non-deterministically pick a transition \((q_{curr}, a, q_{next}) \in \delta\).
    \item If \(a = \epsilon\), set \(q_{curr}\) to \(q_{next}\) and continue.
    \item Else if \(a = \Open{x}\) for some variable \(x\) (that is not
      coalesced), do as follows. If \(x \in \dom(\mu')\), then check that \(a =
      a_{i_{curr}}\), then find the position \(i_{close}\) such that
      \(a_{i_{close}} = \Close{x}\). Else if \(x \not\in \dom(\mu')\), guess
      \(i_{close} \geq i_{curr}\) and set \(q_{next}\) to the state reached by
      following the \(\Close{x}\)-transition from the current \(q_{next}\). Do
      the following two things in parallel:
      \begin{itemize}
      \item Set \(i_{curr}\) to \(i_{close}\), \(q_{curr}\) to \(q_{next}\), and continue.
      \item Set \(i_{end}\) to \(i_{close}\), \(x_{curr}\) to \(x\), increment
        \(i_{curr}\), and go to step 2.
      \end{itemize}
    \item Else if \(a\) is \(a = a_{i_{curr}}\), then set \(q_{curr}\) to
      \(q_{next}\) and increment \(i_{curr}\).
    \item Otherwise, reject.
    \end{enumerate}
  \item If \(i_{curr} = i_{end}\), accept.
\end{enumerate}

Now we will sketch a proof of correctness. By the definition of the semantics of
rules, it is clear that there is a correspondence between mappings and a set of
runs for the automata that compose the rule. It is easy to see that the
algorithm described above will find accepting runs for each of the automatons
that correspond to variables instanced by the rule. These runs will correspond
to a mapping \(\nu\) which is an extension of \(\mu'\) and that can be easily be
transformed into an extension of \(\mu\) by separating the coalesced variables.
To see why the algorithm will accept whenever such a \(\nu\) exists, consider
the following. It can be proven without much difficulty that, given the nested
structure of tree-like rules and the plainness of sequential \(\SRGX\), the way
in which we ordered the variable operations in \(d'\) is the only way in which
they might be actually seen. The only case in which this is not true, is in the
case of indistinguishable variables, which we handled as a separate case.
Therefore, the algorithm will accept whenever there exists an extension to
\(\mu\) that satisfies \(\varphi\).
\hfill\qed\medskip

\subsection*{Proof of Theorem~\ref{theo:RGX-FPT}}


We know that \(\RGX\) can be transformed into equivalent variable automata in
polynomial-time. Therefore, we will only consider that case.

Let \(A\) be a variable automaton, \(d\) a document, \(\mu\) a mapping and \(k\)
the number of variables in \(A\), that is, \(k = |Var(A)|\). We can decide this
instance of the \(\EVAL[\VA]\) problem using the same reduction from the proof
of Theorem~\ref{theo:ptime-sequential-sVA}, but with two modifications.

First, we change the algorithm that decides if \((p, T_i, q) \in \delta'\), for
some given states \(p, q \in Q\) and \(i \in [1, n + 1]\). The original
algorithm will not work in this case because \(A\) might not be sequential.
Thus, now we iterate over all possible total orders over the set \(T_i\) (there
are \(|T_i|!\) such orders) and let \((t_1, \ldots, t_{|T_i|})\) be a sequence
with the elements of \(T_i\) according to that order. We give \((t_1, \ldots,
t_{|T_i|})\) as an additional input to the algorithm and proceed in a similar
way than before, but we keep an additional counter \(e\) with the current
position in the new sequence (we set \(e\) to 1 at the start). Whenever the
algorithm chooses a transition with a variable operation, it compares it with
\(t_e\): if it is the same, it increments \(e\); otherwise, it rejects. At the
target state \(q\) we accept if and only if \(e = |T_i| + 1\), which means we
saw all the variable operations of \(|T_i|\) exactly once. Notice that this
gives an algorithm that runs in time at most \(k! p(n)\), where \(p\) is a
polynomial.

Second, we slightly change the way we handle a variable \(x\) when \(x \not\in
\dom(\mu)\). Instead of replacing the variable operation transitions of \(x\)
with \(\epsilon\)-transitions, we preserve them as they are. In this part of the
algorithm, we will iterate over all valid sequences of variable operations in
\(\{\Open{x}, \Close{x} \mid x \in (\Var(A) \;\setminus\; \dom(\mu))\}\). We say
that a sequence of variable operations is \textit{valid} if, for every variable
\(x\): (1) the operations \(\Open{x}\) and \(\Close{x}\) appear at most once;
(2) if \(\Close{x}\) is in the sequence, then \(\Open{x}\) is in the sequence at
an earlier position. For example, \(\Open{x}, \Open{y}, \Close{x}, \Close{y}\)
and \(\Open{x}, \Open{z}, \Close{x}, \Open{y}\) would be two valid sequences of
operations for variables \(x, y, z\). Given a sequence of operations, the
modified automaton, and the modified document, the problem then reduces to
checking if the final state of the variable automaton is reachable from its
initial state, subject to the constraint that the chosen transitions must match
the sequence of operations and the document.

Formally, the algorithm would be the following. Let \(A' = (Q, q_0, q_f,
\delta')\) be the modified variable automaton, and let \(d' = a_1a_2 \cdots a_n\)
be the modified input document (the label). Throughout the algorithm we will
keep: the current position in the document, \(i_{doc}\); the current position in
the sequence of operations, \(i_{seq}\); and the current state \(q_{curr}\). For
every valid sequence of operations \(s_1, s_2, \ldots, s_m\) we proceed as
follows:
\begin{itemize}
\item Set \(q_{curr}\) to \(q_0\), \(i_{doc}\) to 1, and \(i_{seq}\) to 1.
\item While \(q_{curr} \neq q_f\):
  \begin{itemize}
  \item Non-deterministically pick a transition \((q_{curr}, a, q_{next}) \in
    \delta\) such that \(a = a_{i_{doc}}\) or \(a = s_{i_{seq}}\). If no such
    transition exists, then reject.
  \item Set \(q_{curr}\) to \(q_{next}\), and if \(a = a_{i_{doc}}\), increment
    \(i_{doc}\) by one; otherwise, increment \(i_{seq}\) by one.
  \end{itemize}
\item if \(i_{doc} = n + 1\), then accept; otherwise, reject.
\end{itemize}
If at any point the counters go ``out of bounds'', then we also reject. This
part of the algorithm will run in time at most \((2k)! q(n)\), for some
polynomial \(q\).

It is straightforward to prove that these modification will not alter the
correctness of the algorithm. Also, by combining the different parts of the
algorithm, we will get a total running time of \(k! p(n) + (2k)! q(n) + r(n)\)
where \(p, q, r\) are polynomials. This is in \(O(f(k) n^{c})\) for some
constant \(c\) and some function \(f\). Therefore, the problem is in \(\FPT\)
(\cite{Downey1999}).
\hfill\qed\medskip


\section{PROOFS FROM SECTION~\secref{sec:complexity}}

\subsection*{Proof of Theorem~\ref{theo:sat-general}}


First, we will prove that \(\sat[\VA]\) is in \(\NP\). In order to do this, we
prove a lemma that will limit the size of the documents we must consider.
\begin{lemma}
  Let \(A = (Q, q_0, q_f, \delta)\) be a \(\VA\). If \(A\) is satisfiable, then
  there exists a document of size at most \((2|\VV| + 1)|Q|\) that satisfies it.
\end{lemma}
The proof of this lemma follows a similar idea to the idea behind the pumping
lemma for regular languages (\cite{HopcroftU79}). Suppose the smallest document
\(d = a_1 \cdots a_n\) that satisfies \(A\) is of size greater than \((2|\VV| +
1)|Q|\), and let \(\mu\) be its corresponding mapping. Then, there must exist a
substring \(a_k \cdots a_l\) in \(d\) of size at least \(|Q| + 1\) inside which
\(A\) does not use any variable operations (since \(A\) can use at most
\(2|\VV|\) variable operations). Denote the state of \(A\) after processing
\(a_i\) as \(q_i\). Since \(A\) has \(|Q|\) states, there must exist \(i, j \in
[k, l]\) such that \(i < j\), \(q_i = q_j\), and \(|a_k \cdots a_ia_{j + 1}
\cdots a_l| \leq |Q|\). Because \(A\) does not use any variable operations in
this substring, it is clear that if \(A\) accepts \(d\) and \(\mu\), then it
will accept \(d' = a_1 \cdots a_i a_{j+1} \cdots a_n\) and \(\mu'\), where
\(\mu'\) is \(\mu\) with all the positions greater than \(j\) adjusted by \((j -
i)\). If we repeat this for all substrings of size greater than \(|Q|\) with no
variable operations, then the final document will have size at most \((2|\VV| +
1)|Q|\), contradicting our initial supposition. This proves the lemma.

A direct consequence of the previous lemma is that every satisfiable \(\VA\)
\(A\) has an accepting run that is at most polynomial in size with respect to
\(A\). Therefore, a \(\NP\) algorithm for \(\sat[\VA]\) is to simply guess a run
and check that it is an accepting run (which can easily be done in
polynomial-time).

Now, we prove that \(\sat[\SRGX]\) is \(\NP\)-hard. Notice that this implies
that \(\sat[\VA]\) and \(\sat\) of extractions rules are also \(\NP\)-hard.
Consider the proof of Theorem~\ref{theo:eval-hardness-general}. Notice that the
expression \(\gamma_\alpha\) is satisfiable if and only if it is satisfied by
document \(d = \varepsilon\), since \(\gamma_\alpha\) matches only empty
documents. Therefore, \(\ONETHREESAT\) can be reduced to \(\sat[\SRGX]\). Since
the former is \(\NP\)-hard, the latter is also \(\NP\)-hard. \hfill\qed\medskip

\subsection*{Proof of Theorem~\ref{theo:sat-sequential}}


Let \(A = (Q, q_0, q_f, \delta)\) be a sequential variable automata. Notice that
any sequential path from \(q_0\) to \(q_f\) corresponds to an accepting run,
because sequential paths respect the correct use of variables. Since \(A\) is
sequential, finding an accepting run for \(A\) is as easy as finding a path from
\(q_0\) to \(q_f\). This problem is equivalent to the problem of reachability on
graphs, which is in \(\NLOGSPACE\).
\hfill\qed\medskip

\subsection*{Proof of Theorem~\ref{theo:sat-rules}}


Consider the proof of Theorem~\ref{theo:eval-daglike-npcomplete}. Notice that the
rule \(\varphi\) in this proof is satisfiable if and only if it is satisfied by
the document \(d = \#\), since \(\varphi\) matches only one \(\#\) symbol.
Therefore, \(\ONETHREESAT\) can be reduced to \(\sat\) of functional dag-like
rules in polynomial time. Since the former problem is \(\NP\)-hard, the latter
must also be \(\NP\)-hard. \hfill\qed\medskip

\subsection*{Proof of Theorem~\ref{theo:contain-general}}


As previously stated, it is easy to see that regular expressions are a subset of
\(\RGX\), and it is known that the containment problem for regular expressions
is \(\PSPACE\)-hard. Therefore, we will only prove that \(\containment[\VA]\) is
in \(\PSPACE\).

Let \(A_1 = (Q_1, q_1^0, q_1^f, \delta_1)\) and \(A_2 = (Q_2, q_2^0, q_2^f,
\delta_2)\) be two variable automata. We will prove that deciding if
\(\semd{A_1} \subseteq \semd{A_2}\) for every document \(d\) is in \(\PSPACE\)
by describing a non-deterministic algorithm that decides its complement. The
algorithm will attempt to prove that there exists a counterexample, that is, a
document \(d\) and a mapping \(\mu\) such that \(\mu \in \semd{A_1}\) and \(\mu
\not\in \semd{A_2}\). At every moment, we will have sets \(S_1 \subseteq Q_1\)
and \(S_2 \subseteq Q_2\) that will hold the possible states in which \(A_1\)
and \(A_2\) might be. We will also have sets \(V\) and \(Y\) which will hold the
available and open variables respectively.

Assume, without loss of generality, that \(\VV = \Var(A_1) = \Var(A_2)\) and
\(\OO = \Op(A_1) = \Op(A_2)\). We define the \(\epsilon\)-closure of a state
\(q\), denoted \(E(q)\), as the set of states reachable from \(q\) by using only
\(\epsilon\)-transitions (including \(q\)). Similarly, we define \(S(q, a) =
\{q' \mid (q, a, p) \in \delta \text{ and } q' \in E(p)\}\), where \(a \in
(\Sigma \cup \OO)\) and \(\delta\) is the relevant transition relation. Given a
set of states \(R\), we define \(E(R) = \bigcup_{q \in R} E(q)\) (and \(S(R)\)
analogously). Lastly, we define \(S(R, aw) = S(S(R, a), w)\), where \(w \in
(\Sigma \cup \OO)^*\).

The algorithm proceeds as follows:
\begin{enumerate}
\item Set \(S_1\) to \(E(q_1^0)\), set \(S_2\) to \(E(q_2^0)\), set \(V\) to
  \(\VV\), and set \(Y\) to \(\emptyset\).
\item If \(q_1^f \in S_1\) and \(q_2^f \not\in S_2\), then accept. Otherwise,
  guess either an element \(a\) from \(\Sigma\) or a set of variable operations
  \(P \subseteq \OO\).
\item If the algorithm guessed \(a \in \Sigma\) then:
  \begin{enumerate}
  \item Set \(S_1\) to \(S(S_1, a)\) and \(S_2\) to \(S(S_2, a)\).
  \item Go to step 2.
  \end{enumerate}
\item If the algorithm guessed a set \(P\) of variable operations, then:
  \begin{enumerate}
  \item Check that \(P\) is compatible with \(V\) and \(Y\). If they are, the
    update \(V\) and \(Y\) accordingly; if not, reject.
  \item Let \(\Perm(P)\) be the set of all strings that are permutations of \(P\).
  \item Set \(S_i\) to \(\bigcup_{w \in \Perm(P)} S(S_i, w)\) for \(i \in \{1, 2\}\).
  
  \item Go to step 2.
  \end{enumerate}
\end{enumerate}
It is clear that this algorithm uses only polynomial-space, since we are only
guessing strings of polynomial size, and storing information about variables and
states.

Now we prove that the algorithm is correct. Notice that if the algorithm
accepts, then there exists strings \(w_1\) and \(w_2\) differing only on the
ordering of consecutive variable operations, such that \(q_1^f \in S(E(q_1^0),
w_1)\) and \(q_2^f \not\in S(E(q_2^0), w_2)\). Moreover, \(q_1^f \in S(E(q_1^0),
w_1)\) if and only if there exists a document \(d\) and mapping \(\mu\) such
that \(\mu \in \semd{A_1}\). Since \(w_1\) and \(w_2\) generate the same
document-mapping pairs, and the algorithm tries all the possible permutations of
consecutive variable operations,, it is clear that there is no accepting run in
\(A_2\) with label \(w_2\). Therefore \(\mu \not\in \semd{A_2}\).
\hfill\qed\medskip

\subsection*{Proof of Proposition~\ref{prop:determinize}}


Let \(A = (Q, q_0, q_f, \delta)\) be a variable automaton. We will determinize
\(A\) by using the classical method of \textit{subset construction}. Without
loss of generality, we will allow a set of final states instead of a single
final state. We reuse the definitions of \(E(q)\) and \(S(q)\) from the proof of
Theorem~\ref{theo:contain-general}.

We define the deterministic variable automaton \(A^{det} = (Q', q_0', F',
\delta')\) as follows. Let \(Q' = 2^Q\), \(q_0' = E(q_0)\), \(F' = \{ P \in Q'
\mid q_f \in P \}\). The transition \((P, a, P') \in \delta'\) if and only if
\(P' = \bigcup_{q \in P} S(q, a)\).

Now we will prove that for every document \(d\) and mapping \(\mu\), \(\mu \in
\semd{A}\) if and only if \(\mu \in \semd{A^{det}}\). Let \(\rho\) be an
accepting run for \(d\) and \(\mu\) on \(A\). Then it is easy to prove by
induction that \(\rho\) can be mapped to an accepting run \(\rho'\) in
\(A^{det}\). For the base case, we have that \(q_0 \in q_0'\). For the inductive
case, consider that \(\rho\) uses transition \((p, a, p')\), and that the last
state we appended to \(\rho'\) is \(P\): if \(a = \epsilon\) then \(p' \in P\)
and we do nothing to \(\rho'\); if \(a \in (\Sigma \cup Op)\) then there exist
\((P, a, P') \in \delta'\) such that \(p' \in P'\), so we add \(P'\) to
\(\rho'\). Since \(\rho'\) uses the same transitions as \(\rho\) (except for
\(\epsilon\)-transitions), \(A^{det}\) will also accept \(d\) and \(\mu\).

Now consider the opposite direction: if there is an accepting run \(\rho'\) in
\(A^{det}\), then there is an accepting run \(\rho\) in \(A\). This is also
easily proved with induction. In this case the inductive hypothesis is that if
there exists a path from \(P\) to \(P'\) using a certain sequence of symbols and
variable operations, then for all \(p' \in P'\) there exists \(p \in P\) such
that there is a path from \(p\) to \(p'\) using the same sequence of symbols and
operations. For the base case we have that \(E(q_0) = q_0'\), so it is trivial.
For the inductive case, consider that \(\rho'\) uses transition \((P, a, P')\).
Consider some state \(p' \in P'\). By definition, there is some state \(q \in
P'\) such that \(p' \in E(q)\) and there exists a state \(p \in P\) such that
\((p, a, q) \in \delta\). By composing the different paths between states, we
get the path that proves our hypothesis. By considering the last state in
\(\rho'\) then, we can build an accepting run \(\rho\).
\hfill\qed\medskip

\subsection*{Proof of Theorem~\ref{theo:contain-det}}


Let \(A_1 = (Q_1, q_1^0, q_1^f, \delta_1)\) and \(A_2 = (Q_2, q_2^0, q_2^f,
\delta_2)\) be deterministic variable automata. Assume, without loss of
generality, that \(\OO =\Op(A_1) = \Op(A_2)\) and \(\VV = \Var(A_1) =
\Var(A_2)\). We will prove the theorem by showing that the complement of this
problem is in \(\Sigma_2^P\). We describe an algorithm that will accept if
there exists a document \(d\) and mapping \(\mu\) such that \(\mu \in
\semd{A_1}\) and \(\mu \not\in \semd{A_2}\). We will use the fact that when we
fix some linear order \(\preceq_{\Op{}}\) over the variable operations, then
there is a unique label \(\lambda\) to each document-mapping pair \((d, \mu)\),
denoted \(L(d, \mu, \preceq_{\Op{}})\).

First, we guess a document \(d\), a mapping \(\mu\), and a linear order
\(\preceq_{\Op{}}^1\) over \(\OO\). Then, for all linear orders
\(\preceq_{\Op{}}^2\) over \(\OO\), we compute the label \(\lambda_1 = L(d,
\mu, \preceq_{\Op{}}^1)\) and the label \(\lambda_2 = L(d, \mu,
\preceq_{\Op}^2)\), and finally, we check if there is a run in \(A_i\) that has
\(\lambda_i\) as a label, for \(i \in \{1, 2\}\). This is equivalent to
checking if a deterministic finite automaton accepts a word, and therefore it
can be done in polynomial time. If \(A_1\) accepts \(\lambda_1\) and \(A_2\)
rejects \(\lambda_2\), then we accept; otherwise, we reject.

It is straightforward to prove that the algorithm is correct. Therefore, it
only remains to show that the guessed document \(d\) is of polynomial size
(since that will determine the size and running time of the rest). This can be
done by using the same ``pumping lemma'' argument from the proof of
Theorem~\ref{theo:sat-general}. In this case, the substrings without variable
operations will be of size at most \(|Q_1| \cdot |Q_2|\); if its longer, then
there are indices \(i\) and \(j\) such that the pair of states of \(A_1\) and
\(A_2\) will be the same at position \(i\) and \(j\), and therefore we can
shorten the substring by removing the characters in between. Therefore, we only
need to consider documents of size at most \((2|\VV| + 1)|Q_1||Q_2|\).

Now we prove that for deterministic sequential variable automata \(A_1, A_2\)
the problem is in \(\coNP\). As in the previous case, we show that the
complement of the problem is in \(\NP\). To do this, we guess a document \(d\)
and a mapping \(\mu\) and then check that \(\mu \in \semd{A_1}\) and \(\mu
\not\in \semd{A_2}\). This is the \(\mdcheck\) problem, which is a special case
of the \(\EVAL\) problem, and since \(A_1\) and \(A_2\) are sequential,
Theorem~\ref{theo:ptime-sequential-sVA} guarantees that we can check this in
polynomial time. The same argument made in the previous case for the size of
the document \(d\) applies here.

It only remains to prove that \(\containment\) of deterministic sequential
variable automata is \(\coNP\)-hard. For this we will describe a polynomial-time
reduction from the \textit{disjunctive normal form validity} problem. The
problem consists in determining whether a propositional formula \(\varphi\) in
disjunctive normal form is valid, that is, all valuations make \(\varphi\) true.
We may assume, without loss of generality, that every clause in \(\varphi\) has
exactly three literals. This problem is known to be \(\coNP\)-complete, since it
can be easily shown to be the complement of the \textit{conjunctive normal form
  satisfiability} problem.

Let \(\varphi = C_1 \vee \cdots \vee C_m\) be a propositional formula in
disjunctive normal form with propositional variables \(\{p_1, \ldots, p_n\}\),
and let \(C_i = l_{i, 1} \wedge l_{i, 2} \wedge l_{i, 3}\) (\(i \in [1, m]\)),
where each \(l_{i, j}\) is a literal. We will describe the procedure for
constructing automata \(A_1 = (Q_1, q_1^0, q_1^f, \delta_1)\) and \(A_2 = (Q_2,
q_2^0, q_2^f, \delta_2)\). The construction will only use variable operation
transitions so, in order to simplify the construction, we use transitions of the
form \((p, x, q)\) to represent a ``gadget'' that opens and closes variable
\(x\) in succession, that is, a new state \(r\) and the transitions \((p,
\Open{x}, r)\) and \((r, \Close{x}, q)\). For the automata, we are going to use
variables \(p_1, \ldots, p_n\) to represent positive literals; \(\overline{p_1},
\ldots, \overline{p_n}\) to represent negative literals; and \(c_1, \ldots,
c_m\) to represent clauses. Thus, we have a total of \(2n + m\) variables.

The automaton \(A_1\) is will consist of a long chain with two parts. In the
first part, states are joined with two parallel transitions \(p_i\) and
\(\overline{p_i}\), for every propositional variable \(p_i\). This forces the
automaton to choose a valuation for the propositional variables. The second part
consists of a path with all the clause variables \(c_i\). This will make the
automaton compatible with \(A_2\). Formally, \(A_1\) is defined as follows:
\begin{gather*}
  Q_1 = \{r_1, \ldots, r_{n + m + 1}\} \qquad
  q_1^0 = r_1 \qquad
  q_1^f = r_{n + m + 1} \\
  \delta_1 = \{(r_i, p_i, r_{i + 1}), (r_i, \overline{p_i}, r_{i + 1}) \mid i
  \in [1, n]\} \;\cup\; \{(r_{n + i}, c_i, r_{n + 1 + i}) \mid i \in [1, m]\}
\end{gather*}

The automaton \(A_2\) will consist of \(m\) independent branches, each one
representing a clause. Each branch has three parts. The first part starts with
the clause variable \(c_i\), and then follows with the variables corresponding
with the literals in \(C_1\). In the second part, states are joined with two
parallel transitions \(p_j\) and \(\overline{p_j}\), for every propositional
variable \(p_j\) not used in \(C_i\). The third part consists of a path with all
the clause variables \(c_k\) such that \(i \neq k\). Formally, for \(i \in [1,
m]\) the branch corresponding to clause \(C_i\) in \(A_2\) is defined as
follows:
\begin{gather*}
  Q_{2, i} = \{s_{i, 1}, \ldots, s_{i, n + m + 1}\} \qquad
  q_{2, i}^0 = s_{i, 1} \qquad
  q_{2, i}^f = s_{i, n + m + 1} \\
  \begin{split}
    \delta_{2, i} &= \{(s_{i, 1}, c_i, s_{i, 2}), (s_{i, 2}, l_{i, 1}, s_{i,
      3}),
    (s_{i, 3}, l_{i, 2}, s_{i, 4}), (s_{i, 4}, l_{i, 3}, s_{i, 5})\} \\
    &\,\cup\; \{(s_{i, 4 + j}, p'_j, s_{i, 5 + j}), (s_{i, 4 + j},
    \overline{p'_j}, s_{i, 5 + j}) \mid\\
    &\qquad j \in [1, n - 3] \text{ and } p'_1,
    \ldots, p'_{n - 3} \text{ are the variables not in } C_i\} \\
    &\,\cup\; \{(s_{i, n + 1 + j}, c_j, s_{i, n + 1 + j}) \mid j \in [1, m]\}
  \end{split}
\end{gather*}
Finally, we define \(Q_2 = \bigcup_{i \in [1, m]} Q_{2, i}\) and \(\delta_2 =
\bigcup_{i \in [1, m]} \delta_{2, i}\). We fuse the initial states of each
branch into a single state \(q_2^0\) and fuse the final states of each branch
into a single state \(q_2^f\).

Now we prove that \(\semd{A_1} \subseteq \semd{A_2}\) for every document \(d\)
if and only if \(\varphi\) is valid. First, notice that we need only consider
\(d = \epsilon\), since this is the only document that may satisfy \(A_1\) and
\(A_2\). First, it is easy to see that each mapping \(\mu\) corresponds to a
valuation \(\nu\), namely, by considering \(\nu(p) = 1\) if \(p \in \dom(\mu)\),
and \(\nu(p) = 0\) if \(\overline{p} \in \dom(\mu)\). The automaton \(A_1\) will
accept the set of mappings that correspond to all possible valuations over
\(p_1, \ldots, p_n\). It is also easy to see that the branch \(i\) in \(A_2\)
will accept mapping \(\mu\) if and only if \(\mu\) corresponds to a valuation
that satisfies clause \(C_i\). Therefore, if \(\semd{A_1} \subseteq
\semd{A_2}\), then \(A_2\) accepts the mappings corresponding to all possible
valuations. This means that for each valuation \(\nu\) there is a clause in
\(\varphi\) satisfied by \(\nu\), which means that \(\varphi\) is valid.
\hfill\qed\medskip

 \subsection*{Proof of Theorem~\ref{theo:contain-point-disjoint}}



 Consider $A$, a deterministic sequential $\VA$ that produces point-disjoint
 mappings. Notice that given a document \(d\) and a mapping \(\mu\), such that
 \(\mu \in \semd{A}\), there is exactly one accepting run of \(A\) over \(d\)
 that produces \(\mu\). This follows from \(\mu\) being point-disjoint and \(A\)
 closing all the variables it opens, which means that variable operations can
 only occur in a specific order; and \(A\) being deterministic, which means that
 at every step there is only one choice that \(A\) can take. It is also easy to see
 that the sequence of symbols and variable operations of this run are the same of
 an accepting run of \(d, \mu\) on any other automaton with these properties.

 With this in mind, we describe an algorithm that decides the complement of this
 problem in \(\NLOGSPACE\). That is, given \(A_1\) and \(A_2\), two deterministic
 sequential \(\VA\) that produce point-disjoint mappings, the algorithm accepts
 if there exists a document \(d\) and mapping \(\mu\) such that \(\mu \in
 \semd{A_1}\) and \(\mu \not\in \semd{A_2}\).

 The algorithm simply consists of running \(A_1\) and \(A_2\) in parallel,
 guessing at every step the next transition. If at any moment \(A_1\) is at an
 accepting state and \(A_2\) is not, then we accept. We only need to remember the
 current and next state of \(A_1, A_2\), and the current transition we are
 guessing, all of which takes logarithmic space. Sequentiality guarantees us that
 the runs will always be valid. The argument from the first paragraph guarantees
 us that the algorithm is correct, since the sequence of operations that made
 \(A_1\) accept is the only one that could have made \(A_2\) accept the same
 document and mapping.
 \hfill\qed\medskip




\end{document}